\documentclass[10pt]{article}
\usepackage{graphicx}
\usepackage{fullpage}
\usepackage{amsmath, amssymb, amsthm, bbm}
\usepackage{enumitem}
\RequirePackage{xcolor} 
\usepackage{caption} 
\usepackage{tabularx}
\usepackage{algorithm}
\usepackage[noend]{algpseudocode}
\usepackage{mathtools}
\usepackage{wrapfig}
\usepackage[T1]{fontenc}
\usepackage{hyperref}
\hypersetup{breaklinks, urlcolor=blue, colorlinks, citecolor=green!50!black, linkcolor=blue!80!black}
\usepackage{tcolorbox}

\usepackage{amsthm}
\theoremstyle{plain}
\newtheorem{theorem}{Theorem}[section]

\newtheorem{lemma}[theorem]{Lemma}

\theoremstyle{definition}
\newtheorem{definition}[theorem]{Definition}

\usepackage{subcaption}
\input{preamble}

\title{A Simple and Fast $(3+\varepsilon)$-approximation for \\Constrained Correlation Clustering}
\author{Nate Veldt \\ Department of Computer Science and Engineering \\  Texas A\&M University \\ nveldt@tamu.edu}
\date{}

\begin{document}

\maketitle

\begin{abstract} 
	In \textsc{Constrained Correlation Clustering}, the goal is to cluster a complete signed graph in a way that minimizes the number of negative edges inside clusters plus the number of positive edges between clusters, while respecting hard constraints on how to cluster certain \emph{friendly} or \emph{hostile} node pairs.
	Fischer et al.~\cite{fischer2025faster} recently developed a $\tilde{O}(n^3)$-time 16-approximation algorithm for this problem.
	We settle an open question posed by these authors by designing an algorithm that is equally fast but brings the approximation factor down to $(3+\varepsilon)$ for arbitrary constant $\varepsilon > 0$.
	Although several new algorithmic steps are needed to obtain our improved approximation, our approach maintains many advantages in terms of simplicity. In particular, it relies mainly on rounding a (new) covering linear program, which can be approximated quickly and combinatorially. Furthermore, the rounding step amounts to applying the very familiar \textsf{Pivot} algorithm to an auxiliary graph.
	Finally, we develop much simpler algorithms for instances that involve only \emph{friendly} or only \emph{hostile} constraints.
\end{abstract}

\section{Introduction}
Given a complete unweighted signed graph $G = (V,E^+, E^-)$,
\ccfull{} (\cc{}) seeks to partition $V$ in a way that minimizes the number of positive edges  between clusters (called \emph{positive mistakes}) plus negative edges inside clusters (\emph{negative mistakes})~\cite{bansal2004correlation}.
The problem has been studied extensively from the perspective of approximation algorithms, and variants of the problem have been applied to diverse applications in social network analysis~\cite{veldt2018correlation,yu2024parclusterers,Li2017motifcc}, bioinformatics~\cite{bhattacharya2008divisive,ben1999clustering}, databases~\cite{hassanzadeh2009framework}, computer vision~\cite{kim2011highcc,yarkony2012fast}, and more. See the work of Bonchi et al.~\cite{bonchi2022correlation} for a recent detailed survey.

This paper focuses on \cccfull{} (\ccc{}), a variant of \cc{} in which $G$ is accompanied by a set of \emph{friendly} edges $\friendly$ and \textit{hostile} edges $\hostile$~\cite{vanzuylen2009deterministic,fischer2025faster,kalavas2025towards}. The goal is to minimize mistakes while ensuring that both endpoints of each friendly edge are in the same cluster, and both endpoints of each hostile edge are in separate clusters. Our main contribution is a $\tilde{O}(n^3)$-time $(3+\varepsilon)$-algorithm for \ccc{}, where $n = |V|$, $\varepsilon > 0$ is an arbitrary constant, and $\tilde{O}$-notation hides polylogarithmic factors (i.e., $O(\text{polylog}(n)) = \tilde{O}(1)$). This matches the best approximation factor of $3$~\cite{vanzuylen2009deterministic} up to an arbitrarily small constant, while satisfying desirable properties in terms of speed and simplicity. In particular, the latter 3-approximation relies on a canonical linear programming relaxation that takes $\Omega(n^{7.11})$ time to solve, whereas ours is much faster and can be made completely combinatorial. Our result directly answers an open question of Fischer et al.~\cite{fischer2025faster}, who developed a $\tilde{O}(n^3)$-time 16-approximation for \ccc{}, and asked whether it was possible to improve the approximation factor to 3 with the same runtime.
Before discussing our techniques and contributions in detail, we provide a more in-depth look at prior work on \cc{} and \ccc{}.

\subsection{Prior work}
\begin{figure}[t]
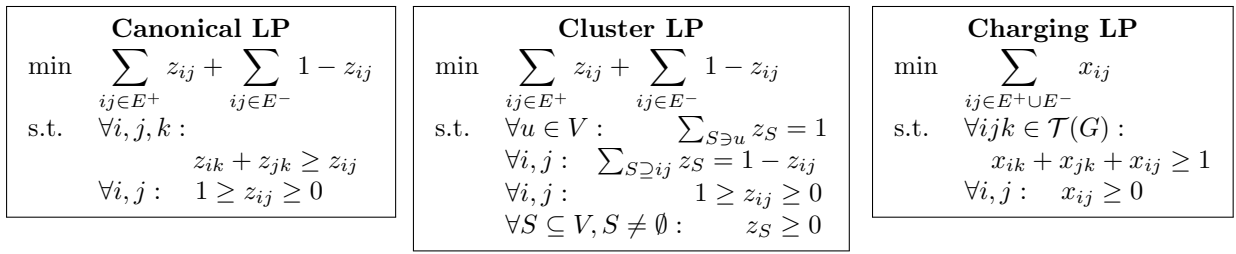

	\centering
	{
		\begin{minipage}[t]{.3\linewidth}
			\begin{align*}
				\boxed{
					\begin{array}{ll}
						\multicolumn{2}{c}{\textbf{Canonical LP}} \\
						\min & {\displaystyle \sum_{ij \in E^+} z_{ij} + \sum_{ij \in E^-} 1 - z_{ij}}  \\
						\text{s.t.} & \forall i,j,k: \\
						&\quad\quad \quad\;\;  z_{ik} + z_{jk} \geq z_{ij}  \\
						& \forall i,j: \quad 1 \geq z_{ij} \geq 0 
					\end{array}
				}
			\end{align*}
		\end{minipage}
		\hfill
		\begin{minipage}[t]{.3\linewidth}
			\begin{align*}
				\boxed{
					\begin{array}{ll}
						\multicolumn{2}{c}{\textbf{Cluster LP}} \\
						\min & {\displaystyle \sum_{ij \in E^+} z_{ij} + \sum_{ij \in E^-} 1 - z_{ij}}  \\
						\text{s.t.} & \forall u \in V:  \hspace{25pt} \sum_{S \ni u} z_S = 1  \\
						& \forall i,j:\hspace{8pt}\sum_{S \supseteq ij} z_S = 1-z_{ij}   \\
						& \forall i,j: \hspace{45pt} 1 \geq z_{ij} \geq 0\\
						&  \forall S\subseteq V, S \neq \emptyset:  \hspace{21pt}  z_{S} \geq 0 
					\end{array}
				}
			\end{align*}
		\end{minipage}
		\hfill
		\begin{minipage}[t]{.3\linewidth}
			\begin{align*}
				\boxed{
					\begin{array}{ll}
						\multicolumn{2}{c}{\textbf{Charging LP}} \\
						\min & {\displaystyle \sum_{ij \in E^+ \cup E^-} x_{ij}}  \\
						\text{s.t.} & \forall ijk \in \badtri{}(G): \\
						&\quad x_{ik} + x_{jk} + x_{ij} \geq 1\\
						& \forall i,j: \quad x_{ij} \geq 0 
					\end{array}
				}
			\end{align*}
		\end{minipage}
	}
	
	\caption{Three LP relaxations for \ccfull{}. $\badtri{}(G)$ represents the set of \textit{bad triangles} in $G = (V,E^+, E^-)$: node triplets $\{i,j,k\}$ where two edges are positive and one edge is negative. Algorithms that rely on solving and rounding variants of the Charging LP tend to be particularly fast and simple.
	}
	\label{fig:lps}
\end{figure}

We first review key algorithmic techniques for standard \cc{} that rely on three different types of linear programming (LP) relaxations for \cc{}, shown in Figure~\ref{fig:lps}. Variants of these LPs have also been very influential in the development of \ccc{} algorithms.

\textbf{Best approximation factors using expensive LPs.}
\ccfull{} was introduced by Bansal et al.~\cite{bansal2004correlation}, who proved NP-hardness and gave a constant-factor approximation. The approximation was later reduced to 4 by rounding the natural Canonical LP (see Figure~\ref{fig:lps}, left)~\cite{CharikarGuruswamiWirth2005}.
The best factor was then improved to 2.5~\cite{AilonCharikarNewman2008} and then 2.06~\cite{ChawlaMakarychevSchrammEtAl2015} by rounding the same LP. 
Recently, a sequence of breakthrough results using more expensive relaxations led to approximation factors below 2~\cite{cohen2022correlation,cohen2023handling,cao2024understanding}, culminating in a $(1.437+\varepsilon)$-approximation. These results can all be described as rounding schemes for a more expensive {Cluster} LP relaxation (Figure~\ref{fig:lps}, middle). The Cluster LP has an exponential number of variables, but prior work has shown how to approximate it in polynomial time~\cite{cao2024understanding} or even sublinear time~\cite{cao2025solving}.
Though significant from a theoretical perspective, algorithms relying on these LP relaxations exhibit disadvantages in terms of simplicity and speed. The Canonical LP has $\Theta(n^3)$ constraints, and takes $\Omega(n^{3\omega})$ time to solve (which is $\Omega(n^{7.11})$ using the current best bound for $\omega$) even using the best theoretical LP solvers.  There is some prior work on practical solvers for the Canonical LP~\cite{veldt2019metric,ruggles2020parallel,sonthalia2022project}, but these have only scaled to graphs with under two hundred thousand nodes. Meanwhile, the sublinear time approximation algorithm for the Cluster LP~\cite{cao2025solving} is an exciting breakthrough result, but is complicated and has a runtime that depends exponentially on a high-degree polynomial in terms of $(1/\varepsilon)$, for a small $\varepsilon > 0$. To our knowledge, no one has attempted to implement any approximate solvers for the Cluster LP.

\textbf{Fast and simple algorithms using the Charging LP.} A significant amount of prior work focuses on faster and simpler algorithms for \cc{} and its variants, at the expense of worse approximation factors~\cite{AilonCharikarNewman2008,veldt2022correlation,makarychev2023single,cohen2024combinatorial,balmaseda2024combinatorial,fischer2025faster}. These often rely in some way on the simpler Charging LP (Figure~\ref{fig:lps}, right) for \cc{}.
The best-known example is the celebrated \textsf{Pivot} algorithm~\cite{AilonCharikarNewman2008}, which obtains an expected 3-approximation for \cc{} by repeatedly selecting a random \emph{pivot} node and clustering it with all its positive neighbors. 
\textsf{Pivot} never solves the Charging LP explicitly, but Ailon et al.~\cite{AilonCharikarNewman2008} proved that the expected cost of this algorithm is within a factor 3 of a feasible solution for the dual of this LP. 

The Charging LP is a covering LP, meaning that it has the form $\min_{x} c^Tx \text{ s.t. } Ax \geq b, x\geq 0$ where $A$, $b$, and $c$ all have nonnegative entries. This special structure comes with many benefits in terms of speed and simplicity.
Fischer et al.~\cite{fischer2025faster} recently used specialized solvers for covering LPs to quickly find approximate solutions for the Charging LP. They used these solvers to develop a {derandomized} version of \textsf{Pivot} that achieves a $(3+\varepsilon)$-approximation for \cc{} in $\tilde{O}(n^3)$ time. There are even simpler and faster techniques for finding a \textit{binary} feasible solution that 3-approximates this LP, by finding a maximal variable-disjoint set of constraints. This has led to extremely fast and simple constant-factor approximations for CC and certain weighted variants~\cite{veldt2022correlation,bengali2023faster,balmaseda2024combinatorial}, which come with simple implementations that scale easily to graphs with millions of nodes. The Charging LP was also used recently to develop the first poly-logarithmic depth parallel algorithm for \cc{} with an approximation factor under 3~\cite{cao2024breaking}.

\textbf{\ccc{}.} Progress on \ccc{} closely mirrors the progression of results for \cc{}. The problem was first studied by van Zuylen and Williamson~\cite{vanzuylen2009deterministic}. These authors developed a 3-approximation by solving and rounding a constrained Canonical LP with constraints $z_{ij} = 0$ for $ij \in \friendly$ and $z_{ij} = 1$ for $ij \in \hostile$. Recently, Kalavas et al.~\cite{kalavas2025towards} considered a similarly constrained variant of the Cluster LP, and showed how to round an approximately optimal LP solution into a better-than-2 approximation for \ccc{}. However, this work was unable to adapt poly-time approximate solvers for the Cluster LP~\cite{cao2024understanding,cao2025solving} to the constrained Cluster LP. Hence, the best approximation factor with a polynomial runtime for \ccc{} is still 3. 
Successfully adapting approximate solvers for the Cluster LP to the constrained case would lead to an improved poly-time approximation factor for \ccc{}. 
However, this approach would exhibit the same disadvantages (with respect to simplicity and implementation) associated with solving the standard Cluster LP. 

Recently, Fischer et al.~\cite{fischer2025faster} provided a simple approach for converting the input graph $G$ into a new auxiliary graph $\auxgraph$, such that running an $\alpha$-approximate \textsf{Pivot} algorithm on $\auxgraph$ produces an $O(\alpha)$-approximation for the original \ccc{} problem on $G$. They then applied their deterministic $(3+\varepsilon)$-approximate \textsf{Pivot} algorithm (for unconstrained \cc{}) to $\auxgraph$. The overall procedure takes only $\tilde{O}(n^3)$-time, but has an approximation factor of 16. The authors therefore posed the following question:
\begin{tcolorbox}[width=\linewidth,boxsep=3pt,left=2pt,right=1pt,top=1pt,bottom=1pt]
	\textbf{Open Question}~\cite{fischer2025faster}. \textit{Can we improve the approximation factor of \cccfull{} from 16
		to 3 while keeping the running time at $\tilde{O}(n^3)$?}
\end{tcolorbox}
A closely related question is whether we can obtain a 3-approximation algorithm by rounding some type of covering LP with $O(n^3)$ constraints, as prior algorithms for variants of \cc{} that apply this strategy tend to be particularly simple and easy to implement.

\subsection{Our approach and contributions}
This paper settles the open question of Fischer~\cite{fischer2025faster} by presenting a $\tilde{O}(\varepsilon^{-3} n^3)$-time $(3+\varepsilon)$-approximation algorithm for \ccc{}. Our algorithm is based on rounding a new covering LP, and therefore enjoys many benefits in terms of simplicity. Our covering LP generalizes the Charging LP in Figure~\ref{fig:lps} for standard \cc{} by incorporating two new types of constraints. The first type ensures that if two edges are forced to be clustered in the same way because of constraints $\friendly \cup \hostile$, then their LP variables are equal. The second special type of constraint lower bounds LP variables for certain triplets of positive edges where every feasible solution is guaranteed to make at least one mistake. We refer to these as \emph{Hostile Edge And Positives} (\heap{}) constraints; see Section~\ref{sec:alg} for a formal definition. We show that enforcing $O(n^3)$ such constraints is sufficient for our analysis.

After approximately solving our covering LP relaxation, we use its output 
to construct an auxiliary graph $\auxgraph$. We show that applying \textsf{Pivot} to $\auxgraph$ is sufficient to obtain the $(3+\varepsilon)$-approximation. Our algorithm can be made completely combinatorial and deterministic. This approach shares similarities with the work of Fischer et al.~\cite{fischer2025faster}, but we consider a different type of covering LP, apply a different auxiliary graph construction, and then provide a more in-depth approximation analysis to achieve the tighter approximation.

As a final contribution, we present substantially simplified algorithms for instances that only involve friendly constraints (\fcc) or hostile constraints (\hcc). For the latter, we develop a $(3+\varepsilon)$-approximation by rounding a simpler covering LP without \heap{} constraints. For \hcc{}, we do not even need to solve a covering LP, although one is considered implicitly in the analysis. Instead, we construct an auxiliary graph based on information from an edge-disjoint set of \textit{dangerous} triangles (involving two positive edges and a hostile edge), and then run a randomized \textsf{Pivot} algorithm on it. This results in a (randomized) $3$-approximate algorithm that runs in $O(n^3)$ time; note that there is no dependence on a constant $\varepsilon > 0$ nor logarithmic factors in the runtime. 
\section{Technical Preliminaries for \ccc{}} 
We begin with technical preliminaries for \ccc{}, which includes several useful concepts for approximating this problem that were developed by van Zuylen and Williamson~\cite{vanzuylen2009deterministic} and by Fischer et al.~\cite{fischer2025faster}.

\subsection{\ccc{} concepts and terminology}
Let $G = (V,E^+, E^-, \friendly, \hostile)$ be an instance of \ccc{} with $n = |V|$. We denote an edge between two nodes $u$ and $v$ simply by $uv$ or $vu$. Every pair of distinct nodes defines either a positive or negative edge. For a clustering $\mathcal{C}$, let $\mathcal{C}(u)$ denote the cluster assignment for node $u$. The cost of $\mathcal{C}$ on graph $G$ is given by:
\begin{equation*}
	\text{cost}_G(\mathcal{C})  = \sum_{uv \in E^+} \mathbbm{1}( \mathcal{C}(u) \neq \mathcal{C}(v)) + \sum_{uv \in E^-} \mathbbm{1}( \mathcal{C}(u) = \mathcal{C}(v)).
\end{equation*}
We refer to $\friendly \cup \hostile$ as {constrained} edges. The clustering $\mathcal{C}$ is \textit{feasible} if it respects all constraints, meaning $\mathcal{C}(u) = \mathcal{C}(v)$ if $uv \in \friendly$ and $\mathcal{C}(u) \neq \mathcal{C}(v)$ if $uv \in \hostile$. We assume throughout the manuscript that there is at least one feasible clustering. One can quickly check that a feasible clustering exists by finding connected components in the graph $G_{\friendly} = (V, \friendly)$ and ensuring that no hostile edge is inside any component.

\textbf{Supernodes and superedges.} Let $\text{comp}(G_{\friendly})$ denote the number of connected components in $G_{\friendly} = (V,\friendly)$ and $s \colon V \rightarrow \{1,2, \hdots, \text{comp}(G_{\friendly})\}$ map each node in $V$ to the index of the component it belongs to in $G_{\friendly}$. If $s(u) = s(v)$, then $u$ and $v$ must be clustered together. We refer to a collection of nodes with the same label under $s$ as a \emph{supernode}, and refer to $s$ as the supernode indicator function. We denote the supernode of $u \in V$ by:
\begin{equation}
	\label{eq:supernode}
	S(u) = \{v \in V \colon s(v) = s(u) \}.
\end{equation}
Let $\mathcal{S}$ be the collection of supernodes of $G = (V,E^+, E^-, \friendly, \hostile)$. 
We refer to two supernodes $(A,B) \in {\mathcal{S} \choose 2}$ as a \emph{superedge}. Let $E^+(A,B)$ and $E^-(A,B)$ denote positive and negative edges in the superedge, respectively. A superedge $(A,B)$ is \textit{hostile}, denoted by $(A,B) \in \hostile$, if there exists $ab \in A \times B$ such that $ab \in \hostile$.

\subsection{Consistent form}
\label{sec:consistent} 
We focus on instances of \ccc{} that satisfy the following definition.
\begin{definition}[Consistent form]
	\label{def:consistent}
	$G = (V,E^+, E^-, \friendly, \hostile)$ is in {consistent} form if $uv \in E^+$ whenever $s(u) = s(v)$, and $uv \in E^-$ whenever $(S(u), S(v))\in \hostile$.
\end{definition}
If we instead start with an instance that is not in consistent form, we can convert it to consistent form in a way that does not affect approximation guarantees or our asymptotic runtime bound.
\begin{definition}
	\label{def:consmap}
	Let $G_0 = (V,E^+_0, E^-_0, \friendly,\hostile)$ be an instance of \ccc{} that is not necessarily in consistent form. Let $s$ be its supernode indicator function, and $\sigma_0(uv) \in \{+,-\}$ denote the sign of edge $uv$ in $G_0$. The \textit{consistent form} of $G_0$ is a new instance $G = (V, {E}^+, {E}^-, {\friendly}, {\hostile})$ where ${E}^+ = \{uv \colon {\sigma}(uv) = +\}$ and ${E}^- = \{uv \colon {\sigma}(uv) = -\}$ for the edge sign function
	\begin{align*}
		{\sigma}(uv) 
		&= 
		\begin{cases}
			+ & \text{if $s(u) = s(v)$} \\
			- & \text{if $(S(u), S(v)) \in \hostile$} \\
			\sigma_0(uv) & \text{otherwise}.
		\end{cases}
	\end{align*}
\end{definition}
Converting $G_0$ to $G$ takes $O(n^2)$ time. Because this conversion does not change the number of nodes, a $\tilde{O}(n^3)$-time algorithm for clustering $G$ also implies a $\tilde{O}(n^3)$-time algorithm for clustering $G_0$. 
\begin{lemma}
	\label{lem:consistent}
	If  $G = (V,E^+, E^-, \friendly,\hostile)$ is the consistent form of $G_0 = (V,E^+_0, E^-_0, \friendly,\hostile)$, then for every $\alpha \geq 1$, an $\alpha$-approximate clustering for $G$ is an $\alpha$-approximate clustering for $G_0$. 
\end{lemma}
\begin{proof}
	Let $M^- = \{uv \colon s(u) = s(v), uv \in E_0^-\}$ be node pairs that share a supernode but are negative edges in $G_0$. Let $M^+ = \{uv \colon (S(u), S(v)) \in \hostile, uv \in E^+_0\}$ be node pairs that have hostile supernodes but are positive edges in $G_0$. For the \ccc{} objective on $G_0$, every feasible clustering will make a mistake at all edges in $M^+ \cup M^-$. For the \ccc{} objective on $G$, no feasible clustering will make a mistake at these edges, since by construction $M^- \subseteq {E}^+$ and $M^+ \subseteq {E}^-$. 
	
	Let $X = (E^+_0 \cup E^-_0) \setminus (M^+ \cup M^-)$ denote the set of all edges that share the same sign in $G_0$ and $G$. Observe that $G_0$ and $G$ will share the same set of optimal solutions for \ccc{}, since they only differ in terms of edges that are forced to be clustered in a certain way because of $\friendly \cup \hostile$. Let $\mathcal{C}^*$ be an optimal clustering for $G_0$ and ${G}$, and $\mathcal{C}$ be an $\alpha$-approximate clustering for ${G}$. Let $\text{OPT}_X$ be the number of mistakes at edges in $X$ that $\mathcal{C}^*$ makes, and $\text{APP}_X$ be the number of mistakes that $\mathcal{C}$ makes on these edges. In ${G}$, neither clustering makes any mistakes at edges in $M^+ \cup M^-$, so the fact that $\mathcal{C}$ is an $\alpha$-approximation for $G$ implies
	\begin{equation*}
		\text{cost}_{{G}}(\mathcal{C}) = \text{APP}_X \leq \alpha \text{OPT}_X = \alpha\text{cost}_{{G}}(\mathcal{C}^*).
	\end{equation*}
	In $G_0$, both clusterings make mistakes at all edges in $M^+ \cup M^-$, so the cost of $\mathcal{C}$ in $G_0$ is
	\begin{align*}
		\text{cost}_{G_0}(\mathcal{C}) = \text{APP}_X  + |M^+ \cup M^-| \leq \alpha(\text{OPT}_X + |M^+ \cup M^-|) = \alpha \text{cost}_{G_0}(\mathcal{C}^*). 
	\end{align*}
	Hence, $\mathcal{C}$ is also an $\alpha$-approximate clustering for $G_0$.
\end{proof}
The rest of the manuscript focuses on instances of \ccc{} in consistent form. This consistency assumption also follows the original presentation of \ccc{} 
(see Section 4 in~\cite{vanzuylen2009deterministic}).

\subsection{Bad triangles and dangerous pairs.}
A triplet of distinct nodes $\{a,b,c\}$ is  a \emph{bad triangle} 
if it contains exactly two positive edges and one negative edge. Let $\badtri(G)$ denote the set of bad triangles in a signed graph $G = (V,E^+, E^-)$ (Figure~\ref{fig:badtri-dangerous}, top). We write $abc \in \badtri(G)$ to indicate a bad triangle that is \emph{centered} at $b$, meaning that $ab \in E^+$, $bc \in E^+$, and $ac \in E^-$. 
\begin{wrapfigure}{r}{1.5in}
	\centering
	\includegraphics[width=1.0in]{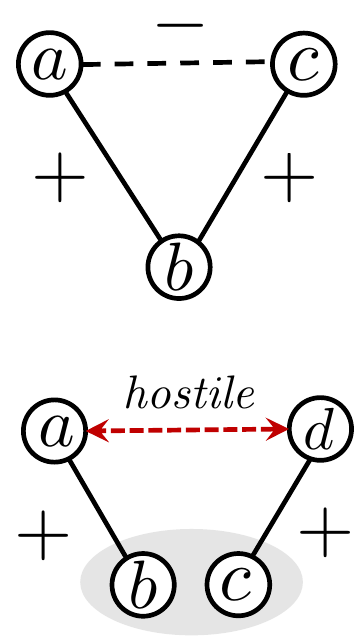}
	\caption{Bad triangle and dangerous pair.}
	\label{fig:badtri-dangerous}
\end{wrapfigure} 
Every clustering must make at least one mistake at each bad triangle.

For \ccc{}, a \emph{dangerous pair} is a pair of positive edges $(ab, cd) \in E^+ \times E^+$ such that $s(b) = s(c)$ and $(S(a), S(d)) \in \hostile$. We refer to this as \textit{dangerous} since applying a pivot procedure to a graph that contains a dangerous pair could result in a hostile edge constraint violation. Every feasible clustering must make a mistake at at least one of these edges. See Figure~\ref{fig:badtri-dangerous}; the gray oval indicates nodes belonging to the same supernode, and a red dashed edge indicates nodes whose supernodes are hostile. If $b = c$, we call this a \emph{dangerous triangle}.

One useful primitive in approximating \ccc{} is computing a maximal edge-disjoint set of dangerous pairs $\mathcal{D}$. By edge-disjoint we mean that each $uv \in E^+$ appears in at most one dangerous pair in $\mathcal{D}$;
 by maximal we mean that we cannot add another dangerous pair to $\mathcal{D}$ without ruining this property. We write $(ab,cd) \in \mathcal{D}$ to indicate a dangerous pair where $s(b) = s(c)$ and $(S(a),S(d)) \in \hostile$. Let
$E_\mathcal{D} = \{ab \in E^+ \colon \exists cd \text{ s.t. } (ab,cd) \in \mathcal{D} \}$.
Since $G$ is in consistent form, $uv \in E_\mathcal{D} \implies s(u) \neq s(v)$. Define $\rho \colon E_\mathcal{D} \rightarrow E_\mathcal{D}$ to be a mapping from each edge $uv \in E_\mathcal{D}$ to its \emph{partner} edge $\rho(uv)$ that satisfies $(uv, \rho(uv)) \in \mathcal{D}$. 

Algorithm~\ref{alg:computeD} is an $O(n^3)$-time algorithm for computing such a set $\mathcal{D}$. This algorithm mirrors a step of the 16-approximation of Fischer et al. (see Lemma 22 of~\cite{fischer2025faster_arxiv}).
\begin{algorithm}[t]
	\caption{\textsf{Compute}$\mathcal{D}(G)$}
	\label{alg:computeD}
	\begin{algorithmic}[1]
		\State \textbf{Input}: \ccc{} instance $G = (V,E^+, E^-,\friendly,\hostile)$ in consistent form
		\State \textbf{Output}: maximal edge-disjoint set of dangerous pairs $\mathcal{D}$
		\State $\mathcal{D} \leftarrow \emptyset$; $E_\mathcal{D} \leftarrow \emptyset$
		\For{$ab \in E^+$ s.t.\ $s(a) \neq s(b)$}
		\State $A = S(a); B = S(b)$
		\If{$ab \in E_\mathcal{D}$}
		\State continue
		\EndIf
		\For{$D \in \mathcal{S} \setminus \{A,B\}$}
		\If{$(A,D) \in \hostile$ \textbf{and} $\exists cd \in E^+(B,D) \setminus E_\mathcal{D}$}
		\State $\mathcal{D} \leftarrow \mathcal{D} \cup (ab, cd)$; $E_\mathcal{D} \leftarrow E_\mathcal{D} \cup \{ab, cd\}$
		\State break
		\EndIf
		\If{$(B,D) \in \hostile$ \textbf{and} $\exists  cd \in E^+(A,D) \setminus E_\mathcal{D}$}
		\State $\mathcal{D} \leftarrow \mathcal{D} \cup (ab, cd)$; $E_\mathcal{D} \leftarrow E_\mathcal{D} \cup \{ab, cd\}$
		\State break
		\EndIf
		\EndFor
		\EndFor
	\end{algorithmic}
\end{algorithm}
It iterates through each $ab \in E^+$ where $s(a) \neq s(b)$, and then iterates through every supernode $D \in \mathcal{S}$. If $D$ is hostile to one of the supernodes $\{S(a),S(b)\}$ but shares a positive edge with the other supernode in $\{S(a), S(b)\}$, the algorithm extracts a dangerous pair involving $ab$ and adds it to $\mathcal{D}$. The resulting set $\mathcal{D}$ is edge-disjoint, since lines 9 and 12 only execute if $ab$ and $cd$ have not yet been added to $E_\mathcal{D}$.
$\mathcal{D}$ is also 
maximal since the algorithm iterates through each $ab \in E^+$ and checks all possible places where a dangerous pair involving $ab$ can be found.

\subsection{Pivoting in an auxiliary graph.}
An existing strategy for approximating \ccc{} is to convert $G$ into an auxiliary graph $\auxgraph$ with a special structure and apply \textsf{Pivot} (Algorithm~\ref{alg:pivot}) to $\auxgraph$~\cite{vanzuylen2009deterministic,fischer2025faster}.
In order for this to work, we must ensure that pivoting in $\auxgraph$ never violates the constraints $\friendly \cup \hostile$ for the original instance $G$. Towards this end, we present the notion of a \emph{pivot-safe} auxiliary graph.
\begin{definition}[Pivot-safe]
	\label{def:pivotsafe}
	Let $G = (V,E^+, E^-, \friendly, \hostile)$ be an instance of \ccc{} with supernode indicator function $s$. An auxiliary graph $\auxgraph = (V, \auxe^+, \auxe^-)$ is {pivot-safe} for $G$ if
	\begin{enumerate}
		\item $s(i) = s(j) \implies $ $ij \in \auxe^+$ and $i$ and $j$ have the same set of positive and negative neighbors in $\auxgraph$. 
		\item $(S(i),S(j)) \in \hostile \implies ij \in \auxe^-$ and there exists no $k$ such that $ik \in  \auxe^+$, $jk \in  \auxe^+$.
	\end{enumerate}
\end{definition}
The first property in Definition~\ref{def:pivotsafe} ensures we do not violate edges in $\friendly$ when applying \textsf{Pivot} to $\auxgraph$. The second property ensures we do not violate edges in $\hostile$. Similar notions of pivot-safe graphs appear in work of van Zuylen and Williamson~\cite{vanzuylen2009deterministic}. Definition~\ref{def:pivotsafe} specifically incorporates the notion of supernodes, which will be convenient for our analysis.

\begin{algorithm}[t]
	\caption{\textsf{Pivot}$(\auxgraph = (V,\auxe^+,\auxe^-))$}
	\label{alg:pivot}
	\begin{algorithmic}[1]
		\State  $\mathcal{C} \gets \emptyset$
		\State $V_1 \gets V$
		\State $i \gets 1$
		\While{$|V_i| > 0$}
		\State  $\auxe^+_i \leftarrow \auxe^+ \cap (V_i \times V_i)$ 
		\State $\auxe^-_i \leftarrow \auxe^- \cap (V_i \times V_i)$ 
		\State $\auxgraph_i \leftarrow (V_i, \auxe^+_i, \auxe^-_i)$
		\State $p = \textsf{ChoosePivot}(\auxgraph_i)$		 
		\State $C_p \gets p \cup \{v\in V_i \colon vp \in \auxe^+_i\}$ \hfill\texttt{// find positive  neighbors}
		\State $\mathcal{C} \gets \mathcal{C} \cup \{C_p\}$  \hfill\texttt{// define new cluster}
		\State $i \gets i + 1$
		\State $V_i \leftarrow V_{i-1} \setminus C_p$		\hfill\texttt{// remove clustered nodes}
		\EndWhile
		\State Return clustering $\mathcal{C}$
	\end{algorithmic}
\end{algorithm}

In Algorithm~\ref{alg:pivot}, the subroutine $\textsf{ChoosePivot}$ represents some strategy for selecting pivots, which could be random (e.g., uniform random pivots) or deterministic. A key contribution of van Zuylen and Williamson was to show deterministic pivoting strategies that led to approximation guarantees for \ccc{} and other variants of \cc{}~\cite{vanzuylen2009deterministic}. We will make use of the following deterministic pivoting strategy and runtime bound for clustering a graph $G = (V,E^+,E^-)$ by applying \textsf{Pivot} to auxiliary graph $\auxgraph = (V,\auxe^+,\auxe^-)$.
\begin{lemma}[Lemma 3.1~\cite{vanzuylen2009deterministic}]
	\label{lem:3pt1}
	\textsf{Pivot} (Algorithm~\ref{alg:pivot}) can be implemented in $O(n^3)$ if in iteration $i$ the pivot node $p$ is chosen to minimize the ratio 
	\begin{equation}
		\label{eq:ratio1}
		\frac{|\badtri^+_p(\auxgraph_i) \cap E^+| + |\badtri^-_p(\auxgraph_i) \cap E^-| }{ \sum_{uv \in \badtri^+_p(\auxgraph_i) \cup \badtri^-_p(\auxgraph_i)} y_{uv}},
	\end{equation}
	where $\{y_{uv} \colon uv \in E^+ \cup E^- \}$ is a set of budgets for edges and where
	\begin{align*}
		\badtri^+_p(\auxgraph_i)  = \{uv \in \auxe^+_i \colon pu \in \auxe_i^+, pv \in \auxe^-_i\} \text{ and } \badtri^-_p(\auxgraph_i)  = \{uv \in \auxe^-_i \colon pu \in \auxe_i^+, pv \in \auxe^+_i\}.
		\end{align*}
\end{lemma}
\begin{proof}
	If we ignore the time it takes to compute the quantities appearing in the numerator and denominator of the ratio in~\eqref{eq:ratio1}, finding a pivot $p$ minimizing this ratio takes $O(n)$ time, and forming the cluster $C_p$ takes $O(n)$ time. Since there are at most $n$ iterations, this part of the computation takes $O(n^2)$ time.
	
	To compute the ratio in~\eqref{eq:ratio1} efficiently, we initialize a list of all bad triangles in $\auxgraph$ before choosing the first pivot.
	For each bad triangle we store the nodes $\{a,b,c\}$ in the triangle, the edge budgets $\{y_{ab}, y_{bc}, y_{ac}\}$ for the triangle, and the edge signs in $G$ and $\auxgraph$. We also maintain a map from each node to the bad triangles it belongs to in $\auxgraph$.
	By iterating through all node triplets in $O(n^3)$ time at the start of the algorithm, we can compute and store all quantities in the numerator and denominator for each node.
	
	In iteration $i$, we choose a pivot node $p$ and form a cluster $C_p$. Then, for each bad triangle that overlaps with $C_p$, we must remove that bad triangle and update the quantities in~\eqref{eq:ratio1} for each node in the triangle. Since we maintain a map from each node to the bad triangles it belongs to, this takes constant time per bad triangle. 
	Since we start with $O(n^3)$ bad triangles and each is removed from the list once, the total maintenance time is $O(n^3)$.
\end{proof}

The 3-approximation for \ccc{} developed by van Zuylen and Williamson~\cite{vanzuylen2009deterministic} works by applying the deterministic pivoting strategy in Lemma~\ref{lem:3pt1} with $y_{uv} = z_{uv}$ where $\{z_{uv}\}$ is an optimal set of variables for the constrained version of the Canonical LP. In our work, the budgets will correspond to variables for an alternative covering LP relaxation, or to indicator variables for certain special edges.

\section{The $(3+\varepsilon)$-approximation for \ccc{}}
\label{sec:alg}
To develop a $\tilde{O}(n^3)$-time $(3+\varepsilon)$-approximation for \ccc{}, we present a new covering LP relaxation and show how to round it by applying \textsf{Pivot} to a carefully constructed auxiliary graph $\auxgraph$. 

\subsection{The covering LP and rounding algorithm}
In the Charging LP for standard \cc{} (Figure~\ref{fig:lps}, top right), the constraint $z_{ij} + z_{jk} + z_{ik} \geq 1$ for $ijk \in \badtri(G)$ reflects the fact that every clustering will make at least one mistake at each bad triangle. Our first technical contribution is a generalized covering LP relaxation for \ccc{} that has $O(n^3)$ constraints and can be approximated in $\tilde{O}(n^3)$-time. Our covering LP involves new types of {superedge} constraints and variables, and a new notion of a \heap{} constraint. The \heap{} constraints depend in turn on first computing a maximal edge-disjoint set of dangerous pairs $\mathcal{D}$, which can be done in $O(n^3)$ time using Algorithm~\ref{alg:computeD}.

\textbf{Superedge variables and constraints.}
Our covering LP involves two variables for each pair of supernodes $A,B \in \mathcal{S}$. Variable $X_{AB}^+ \in [0,1]$ can be viewed as a relaxed binary indicator for whether to make mistakes at all edges in $E^{+}(A,B)$. Variable $X_{A,B}^-$ is a relaxed indicator for making mistakes at $E^-(A,B)$.
We add constraint $X_{AB}^+ + X_{AB}^- \geq 1$, since we either make mistakes at all edges in $E^+(A,B)$ or all edges in $E^-(A,B)$. For each triplet of distinct supernodes $\{A,B,C\}$, every clustering makes mistakes at all edges in $E^+(A,B)$, all edges in $E^+(B,C)$, or all edges in $E^-(A,C)$, so we add constraints of the form
\begin{equation*}
	X_{AB}^+ + X_{BC}^+ + X_{AC}^- \geq 1.
\end{equation*}
Our analysis will often consider LP variables for a specific edge $uv$ and the corresponding supernodes $S(u)$ and $S(v)$. For notational convenience, we define $X_{uv}^+ = X_{S(u)S(v)}^+$ and $X_{uv}^- = X_{S(u)S(v)}^-$.

\textbf{\heap{} constraints.} Given $\mathcal{D}$, Algorithm~\ref{alg:heap} is a procedure for identifying certain triplets of positive 
\begin{wrapfigure}{r}{1.75in} 
	\centering
		\vspace{-5pt}
	\includegraphics[width=1.5in]{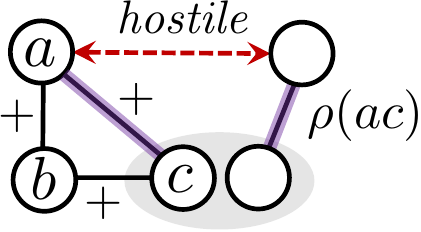}
	\caption{A structure leading to a \heap{} constraint $x_{ab} + x_{bc} + x_{\rho(ac)} \geq 1$.}
	\label{fig:heap}
\end{wrapfigure} 
edges that cannot all be satisfied by a feasible clustering.
This algorithm iterates through each edge $ac \in E_\mathcal{D}$, then iterates through each node $b \notin \{S(a),S(c)\}$.
Consider Figure~\ref{fig:heap} for illustration. 
If $ac \in E_\mathcal{D}$ and $\{ab, bc, ac\}$ is a triangle of positive edges, then every feasible clustering must make a mistake at one of the edges $\{ab, bc, \rho(ac)\}$ to avoid violating one of the constraints in $\friendly \cup \hostile$. We therefore introduce a new LP constraint $x_{ab} + x_{bc} + x_{\rho(ac)} \geq 1$ for this triplet. We refer to this as a \emph{Hostile Edge And Positives} (\heap) constraint. We will need this constraint later to prove a useful bound that will hold if nodes $\{a,b,c\}$ define a bad triangle in the auxiliary graph $\auxgraph$ (see Lemma~\ref{lem:badtri}). Note that there can be at most $O(n^3)$ \heap{} constraints.

\begin{algorithm}[t]
	\caption{\textsf{Find\heap{}s}$(G, \mathcal{D})$}
	\label{alg:heap}
	\begin{algorithmic}[1]
		\State \textbf{Input}: \ccc{} instance $G$ and maximal edge-disjoint set of dangerous pairs $\mathcal{D}$
		\State \textbf{Output}: Set of positive edge triplets defining \heap{} constraints
		\State $\petset \leftarrow \emptyset$
		\For{$ac \in E_\mathcal{D}$}
		\For{$b \in V \setminus \{S(a),S(c)\}$}
		\If{$ab \in E^+$ and $bc \in E^+$}
		\State $\petset \leftarrow \petset \cup \{ab, bc, \rho(ac)\}$
		\EndIf
		\EndFor
		\EndFor
		\State Return $\petset$
	\end{algorithmic}
\end{algorithm}

\textbf{The covering LP relaxation.}
Our new LP relaxation for a \ccc{} instance $G$ is given by
\begin{equation}
	\begin{aligned}
		\label{eq:ccc_lp}
		\min & \displaystyle\sum_{uv \in E^+ \cup E^-} x_{uv}  &\\
		\text{s.t. } & x_{uv} = X_{S(u)S(v)}^+ = X_{uv}^+ & \text{ $\forall  uv \in E^+$}\\
		& x_{uv} = X_{S(u)S(v)}^- = X_{uv}^- & \text{ $\forall uv \in E^-$}\\
				& X_{AA}^+ = 1-X_{AA}^- = 0 & \text{ $\forall  A \in \mathcal{S}$} \\
				& X_{AB}^- = 1-X_{AB}^+ = 0 & \text{ $\forall  (A,B) \in \hostile$} \\
		&\begin{rcases}
			X_{AB}^+ + X_{BC}^+ + X_{AC}^- \geq 1 \\ 
			X_{AB}^+ + X_{BC}^- + X_{AC}^+ \geq 1 \\
			X_{AB}^- + X_{BC}^+ + X_{AC}^+ \geq 1 \\
		\end{rcases} & \text{ $\forall  \{A,B,C\} \in {\mathcal{S} \choose 3}$} \\
		& X_{AB}^+ + X_{AB}^- \geq 1 & \text{$\forall  A,B \in \mathcal{S}$}\\
		& X_{AB}^+ \geq 0, X_{AB}^- \geq 0 & \text{$\forall  A,B \in \mathcal{S}$}	\\
		& x_{ab} + x_{bc} + x_{\rho(ac)} \geq 1 & \text{$\forall  \{ab, bc, \rho(ac)\} \in \petset$.} 
	\end{aligned}
\end{equation}
Theorem~\ref{thm:covering} summarizes two theorems from Fischer et al.~\cite{fischer2025faster}, which are based on previous machinery for covering LPs using combinatorial~\cite{garg1998faster,fleischer2004fast} or non-combinatorial methods~\cite{wang2016unified,allen2019nearly}.
\begin{theorem}[Theorems 13 \& 14 in~\cite{fischer2025faster}]
	\label{thm:covering}
	For every $\varepsilon \in (0,1)$, a $(1+\varepsilon)$-approximate solution for a covering LP with at most $N$ nonzero entries in the constraint matrix can be found in $\tilde{O}(N \varepsilon^{-1})$ time, and can be found in time $\tilde{O}(N \varepsilon^{-3})$ by a combinatorial algorithm.
\end{theorem}
Because LP~\eqref{eq:ccc_lp} is a covering LP with $O(n^3)$ constraints, and each constraint involves at most 3 variables, we know Theorem~\ref{thm:covering} applies with $N = O(n^3)$. This means we can solve the LP in $O(n^3 \varepsilon^{-1})$ time, or in $O(n^3 \varepsilon^{-3})$ time using a simpler combinatorial approach.

\textbf{Building the auxiliary graph $\auxgraph$.} 
\begin{algorithm}[t]
	\caption{\textsf{BuildAuxGraph}$(G = (V,E^+, E^-,\friendly,\hostile),E_\mathcal{D}, \mathcal{X})$}
	\label{alg:buildghat}
	\begin{algorithmic}[1]
		\State \textbf{Input}: \ccc{} instance $G$, edge set $E_\mathcal{D}$, feasible solution $\mathcal{X}$ for LP~\eqref{eq:ccc_lp}
		\State \textbf{Output}: Pivot-safe auxiliary graph $\auxgraph$
		\State $\auxe^+ \longleftarrow \{uv \colon s(u) = s(v) \}$ \hfill \texttt{//initialize positive edge set} 
		\For{$(A,B) \in {\mathcal{S} \choose 2}$}
		\If{$E^+(A,B) \setminus E_\mathcal{D} \neq \emptyset$ and $X_{AB}^+ < \min\{X_{AB}^-, \frac23\}$} \hfill \texttt{//condition for more positive edges} 
		\State $\auxe^+ \longleftarrow \auxe^+ \cup \{uv \colon u \in A, v \in B\}$ 
		\EndIf
		\EndFor
		\State $\auxe^- = {V \choose 2} \setminus \auxe^+$ \hfill \texttt{//all other edges are negative}
		\State Return $\auxgraph = (V, \auxe^+, \auxe^-)$
	\end{algorithmic}
\end{algorithm}
Let $\mathcal{X} = \{X_{AB}^-, X_{AB}^+, x_{uv} \colon A,B \in \mathcal{S}, uv \in E^+ \cup E^-\}$ denote a set of approximately optimal LP variables for LP~\eqref{eq:ccc_lp}. To construct $\auxgraph = (V,\auxe^+, \auxe^-)$, we first ensure that all edges inside supernodes are included in $\auxe^+$. Then, for each pair of distinct supernodes $(A, B) \in {\mathcal{S} \choose 2}$, we ensure edges in $(A,B)$ are either all positive or all negative in $\auxgraph$, by considering two cases (see Figure~\ref{fig:bigfigure}):
\begin{enumerate}
	\item \textbf{Case 1}: If $E^+(A,B) \subseteq E_\mathcal{D}$, make every edge between $A$ and $B$ negative in $\auxgraph$ (Figure~\ref{fig:intermediate}).
	\item \textbf{Case 2}: If there is an edge between $A$ and $B$ from $E^+ \backslash E_\mathcal{D}$, then:
	\begin{itemize}
		\item If $X_{AB}^+ \geq X_{AB}^-$, make every edge between $A$ and $B$ negative in $\auxgraph$.
		\item If $X_{AB}^- > X_{AB}^+$ and $X_{AB}^+ \geq \frac23$, make every edge between $A$ and $B$ negative in $\auxgraph$.
		\item If $X_{AB}^- > X_{AB}^+$ and $X_{AB}^+ < \frac23$, make every edge between $A$ and $B$ positive in $\auxgraph$.
	\end{itemize}
\end{enumerate}
This construction is captured more succinctly by Algorithm~\ref{alg:buildghat}. This always produces a pivot-safe graph.
\begin{lemma}
	\label{lem:pivotsafe}
	Running \textsf{Pivot} on graph $\auxgraph$ produced by \textsf{BuildAuxGraph} (Algorithm~\ref{alg:buildghat}), with any choice of pivot nodes, produces a feasible clustering for the \ccc{} instance $G = (V,E^+, E^-,\friendly,\hostile)$. 
\end{lemma}
\begin{proof}
	We show that $\auxgraph$ satisfies the two properties in Definition~\ref{def:pivotsafe}. The first property follows from the fact that all edges inside supernodes are in $\auxe^+$, and edges between an arbitrary pair of supernodes $A$ and $B$ are either all positive or all negative in $\auxgraph$. 
	For the second property in Definition~\ref{def:pivotsafe}, 
	let $i$ and $j$ be two nodes satisfying $(S(i), S(j)) \in \hostile$ and note that $ij \in \auxe^-$ because $X_{ij}^- = 0$.
	Assume towards a contradiction that there exists some $k$ such that $ik \in \auxe^+$ and $jk \in \auxe^+$. Define $I = S(i)$, $J = S(j)$, and $K = S(k)$. From the construction of $\auxgraph$, we know that $ik \in \auxe^+$ implies there exists some edge $e \in E^+(I,K) \setminus E_\mathcal{D}$. Similarly, $jk \in \auxe^+$ implies there exists some $f \in E^+(J,K) \setminus E_\mathcal{D}$. However, this implies that $\{e,f\}$ is a new dangerous pair we can add to $\mathcal{D}$, contradicting the maximality of $\mathcal{D}$.
\end{proof}

\begin{algorithm}[t]
	\caption{\textsf{ConstrainedCoverRound}$(G = (V,E^+, E^-,\friendly,\hostile), \varepsilon)$}
	\label{alg:mainalg}
	\begin{algorithmic}[1]
		\State \textbf{Input}: \ccc{} instance $G$ in consistent form, approximation parameter $\varepsilon \in (0,1)$
		\State \textbf{Output}: feasible clustering for \ccc{}; a $(3+\varepsilon)$-approximation for certain pivot choices
		\State $\mathcal{D} \longleftarrow \textsf{Compute}\mathcal{D}(G)$ 
		\State $\petset \longleftarrow \textsf{Find\heap{}s}(G, \mathcal{D})$
		\State $\mathcal{X} \longleftarrow $ $\big(1+\frac{\varepsilon}{3}\big)$-approximate solution to LP~\eqref{eq:ccc_lp} 
		\State $\auxgraph \longleftarrow \textsf{BuildAuxGraph}(G, E_\mathcal{D},\mathcal{X})$
		\State $\mathcal{C} \longleftarrow \textsf{Pivot}(\auxgraph)$
		\State Return $\mathcal{C}$
	\end{algorithmic}
\end{algorithm}

\begin{figure}[t!]
	\centering
	\begin{subfigure}[b]{0.29\textwidth}
		\centering
		\includegraphics[width=\textwidth]{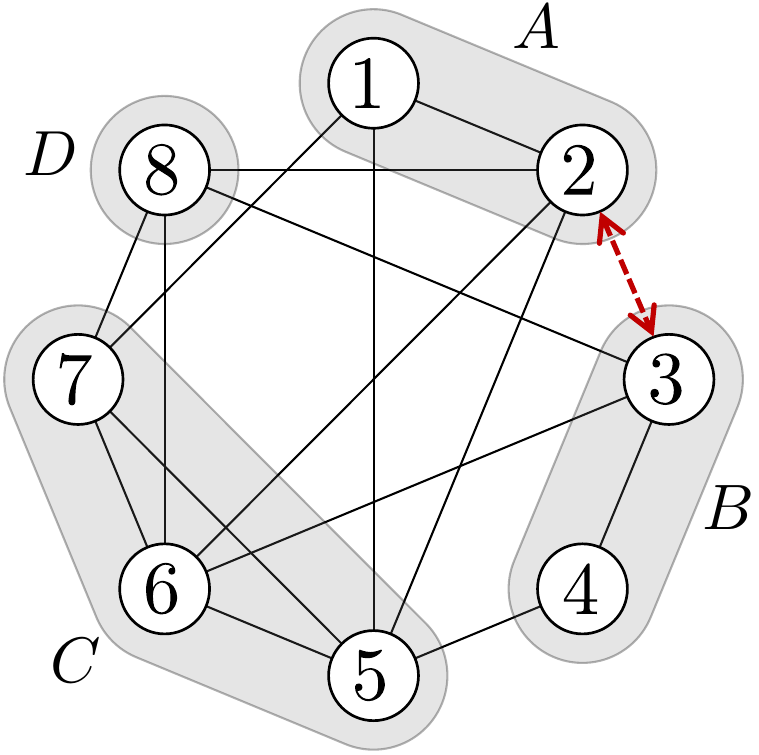}
		\caption{\ccc{} instance $G$}
		\label{fig:instance}
	\end{subfigure}   
	\hfill
	\begin{subfigure}[b]{0.29\textwidth}
		\centering
		\includegraphics[width=\textwidth]{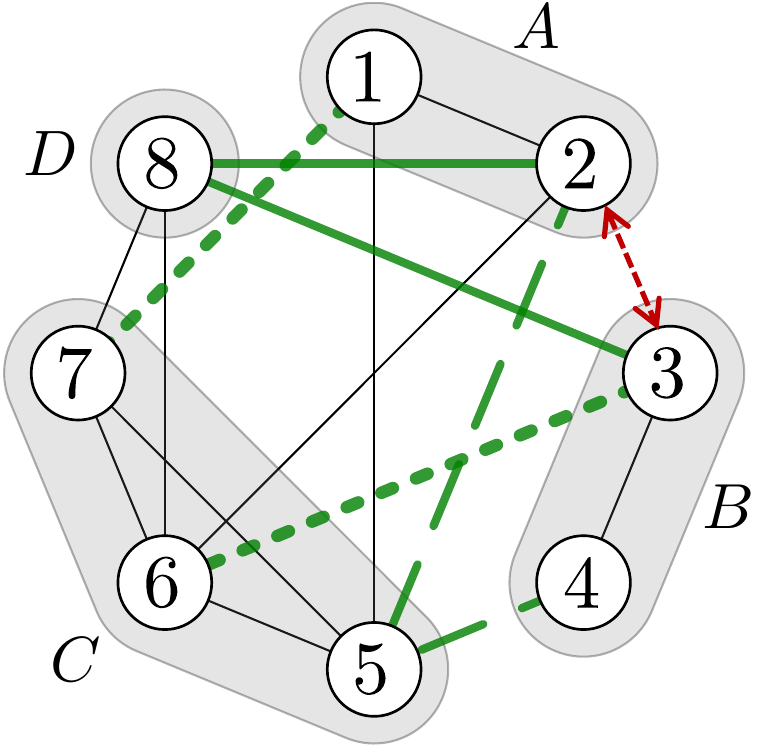}
		\caption{$E_\mathcal{D}$ edges, which come in pairs}
		\label{fig:computeD}
	\end{subfigure}
	\hfill
	\begin{subfigure}[b]{0.34\textwidth}
		\centering
		\includegraphics[width=\textwidth]{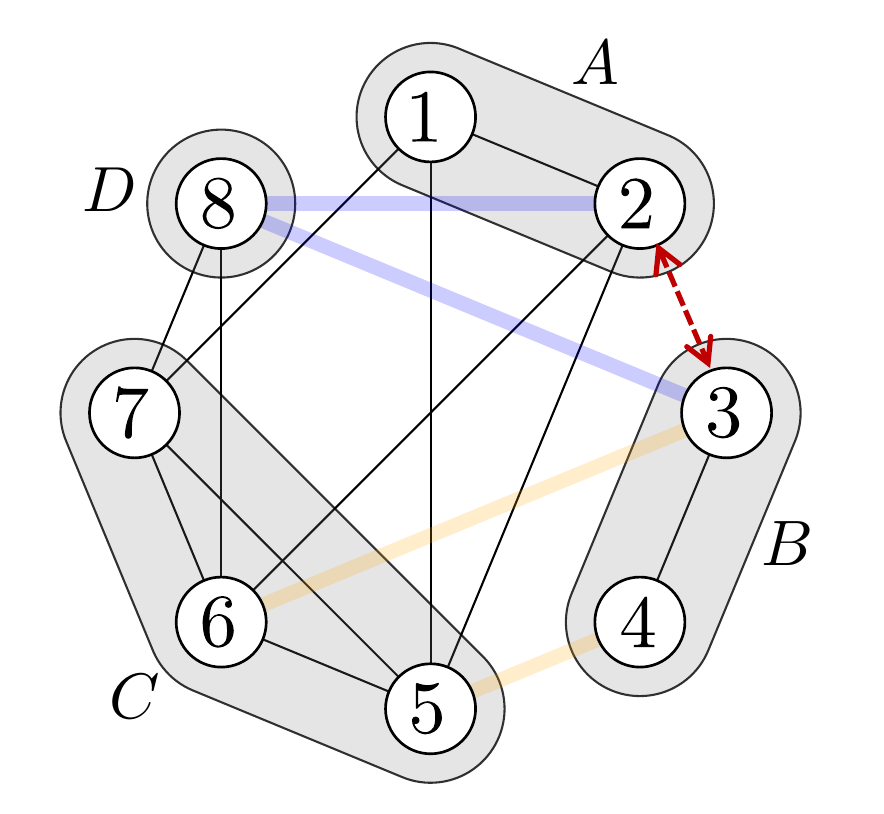}
		\caption{Edges flipped in case 1 of building $\auxgraph$}
		\label{fig:intermediate}
	\end{subfigure}
	\hfill
	\begin{subfigure}[b]{0.32\textwidth}
		\centering
		\includegraphics[width=\textwidth]{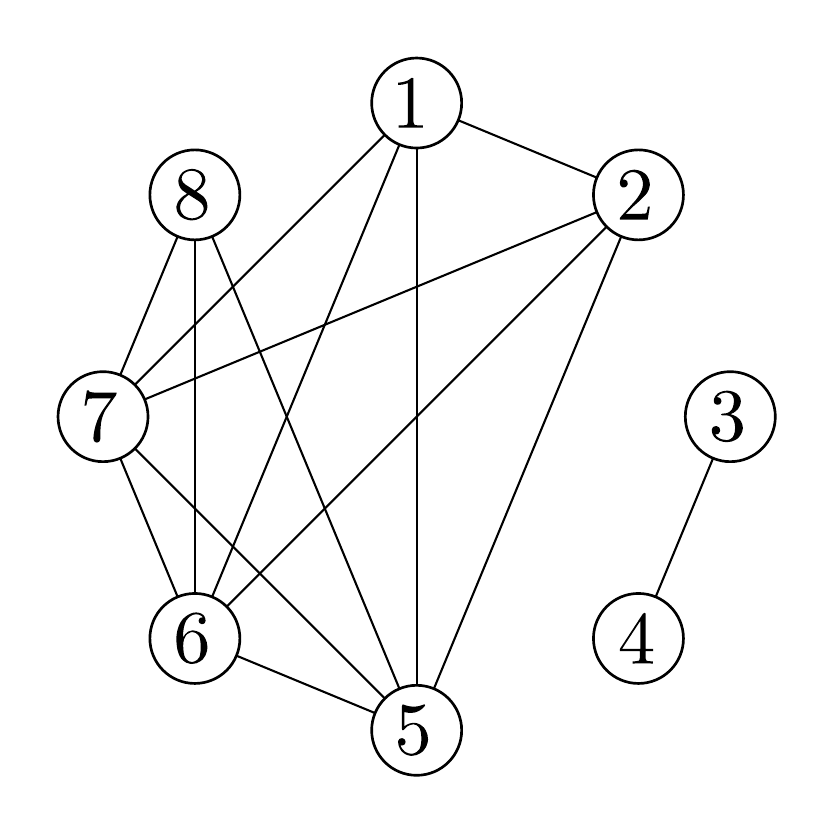}
		\caption{Pivot-safe auxiliary graph $\auxgraph$}
		\label{fig:aux}
	\end{subfigure}
	\hfill
	\begin{subfigure}[b]{0.32\textwidth}
		\centering
		\includegraphics[width=\textwidth]{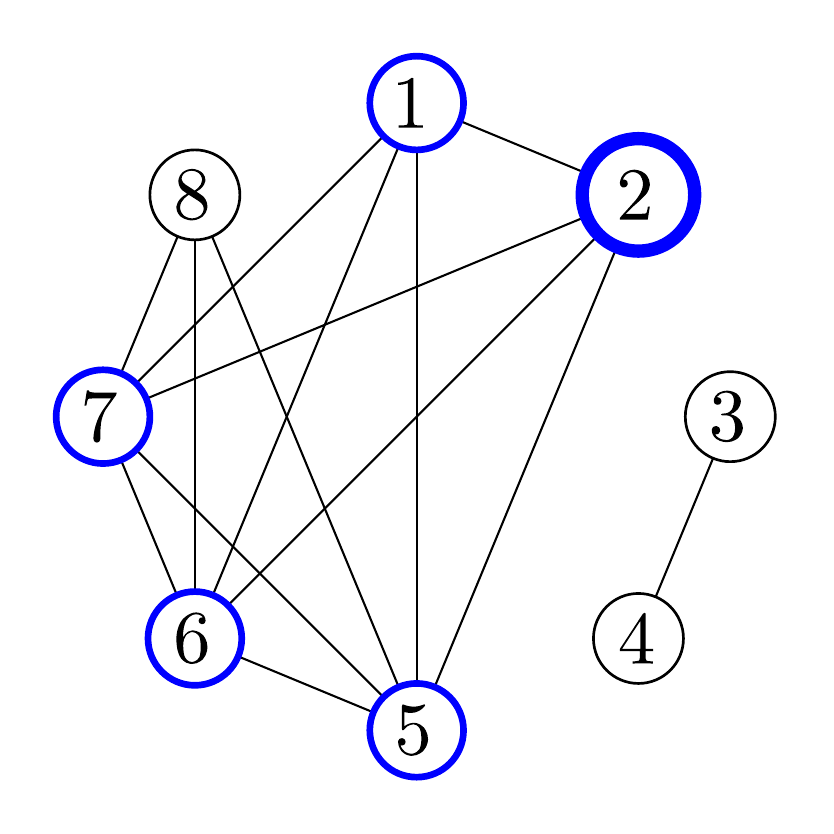}
		\caption{First pivot (node 2) and cluster}
		\label{fig:pivot1}
	\end{subfigure}
	\hfill
	\begin{subfigure}[b]{0.32\textwidth}
		\centering
		\includegraphics[width=\textwidth]{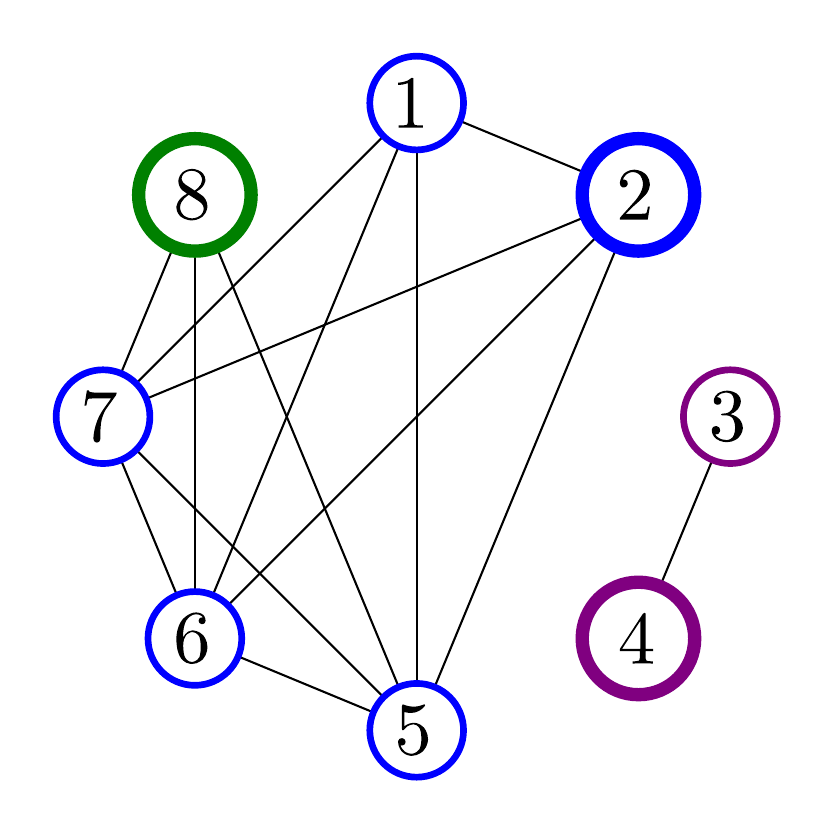}
		\caption{2nd and 3rd  pivots (nodes 8 \& 4)}
		\label{fig:pivot2}
	\end{subfigure}
	\caption{(a)~We display an instance of \ccc{} with four supernodes (gray convex sets) $A = \{1,2\}$, $B = \{3,4\}$, $C = \{5,6,7\}$, and $D = \{8\}$. Solid lines indicate positive edges. For easier visualization, the absence of a line between two nodes indicates a negative edge. Supernodes $A$ and $B$ are hostile since $(2,3) \in \hostile$ (dashed red line). (b)~Green lines indicate a maximal edge-disjoint set of dangerous pairs $\mathcal{D}$, where different line styles (solid, dotted, dashed) are used to indicate how edges are paired. (c)~In constructing $\auxgraph$, four edges from $E_\mathcal{D}$ are flipped from positive to negative because of the first case considered when constructing $\auxgraph$, since $E^+(A,B)$, $E^+(A,D)$, $E^+(B,D)$, and $E^+(B,C)$ are all subsets of $E_\mathcal{D}$. Edges $\{(2,8), (3,8)\}$ are in class $E^\texttt{b}$, and $\{(3,6), (4,5)\}$ are in class $E^\texttt{o}$ (see Section~\ref{sec:classes}). (d)~For remaining pairs of supernodes, edges are determined based on LP values. For this example, $X_{AC}^+ = X_{CD}^+ = 0$, so all edges from $A$ to $C$ and all edges from $C$ to $D$ become positive. The resulting auxiliary graph $\auxgraph$ is pivot-safe. (e)-(f)~For illustration, if node 2 is the first pivot, followed by nodes 8 and then 4, this results in a feasible clustering for instance $G$ with three clusters.}
	\label{fig:bigfigure}
\end{figure}

\textbf{Proof overview and intuition for $\auxgraph$.} Our approximation algorithm for \ccc{} (Algorithm~\ref{alg:mainalg}) rounds LP~\eqref{eq:ccc_lp} by running \textsf{Pivot} on $\auxgraph$. 
To prove an approximation guarantee, we will consider the set of bad triangles in $\auxgraph$ and use a charging argument to prove that the number of mistakes made is within a factor 3 of the LP objective cost $\sum_{uv} x_{uv}$. This holds in expectation if we choose pivots uniformly at random, and holds with certainty if we apply the pivoting strategy in Lemma~\ref{lem:3pt1} with a careful choice of edge budgets. This approach is similar to the framework van Zuylen and Williamson~\cite{vanzuylen2009deterministic} used to prove a deterministic 3-approximation for \ccc{} by rounding the Canonical LP relaxation. However, our covering LP is not as tight as the (more expensive) Canonical LP, so we require several new insights and proof techniques. 

For our analysis, we first partition edges and bad triangles in $\auxgraph$ into different types based on the construction of $\auxgraph$.
We remark that as long as we keep the first case in building $\auxgraph$ (which checks whether $E^+(A,B) \subseteq E_\mathcal{D}$), there are many ways to round other edges that ensure $\auxgraph$ is pivot-safe. However, the threshold $X_{AB}^+ < \min \{X_{AB}, \frac23\}$ is chosen very carefully to balance several competing bounds in Lemmas~\ref{lem:edgefacts} and~\ref{lem:badtri}, all of which are needed for our $(3+\varepsilon)$-approximation guarantee. For these lemmas, we will repeatedly apply that fact that $X_{AB}^+ \geq X_{AB}^-$ implies that every edge between $A$ and $B$ will be negative in $\auxgraph$. This leads to the threshold of the form $X_{AB}^+ < \min \{X_{AB}^-, \beta\}$ in Algorithm~\ref{alg:buildghat}, for deciding whether to make all edges between $A$ and $B$ positive. The bounds in Lemma~\ref{lem:edgefacts} only hold if $\beta \leq \frac23$, and Case 3 of Lemma~\ref{lem:badtri} only holds if $\beta \geq \frac23$, which justifies our choice of $\beta = \frac23$. The bounds in Lemma~\ref{lem:badtri} also rely crucially on the \heap{} constraints in the LP. If we are willing to loosen any of these constraints and bounds, we could obtain a simpler algorithm for \ccc{}, but at the expense of a worse approximation factor. 
The construction of $\auxgraph$ in Algorithm~\ref{alg:buildghat} is designed to be as simple as possible while still ensuring a $(3+\varepsilon)$-approximation guarantee.

\subsection{Edge classes and bad triangle types}
\label{sec:classes}
We say an edge $ij$ is \emph{flipped} if its sign in $G$ differs from its sign in $\auxgraph$. 
In order to succinctly refer to edges based on their sign in $G$ and $\auxgraph$, we define $E^{\texttt{-}\texttt{-}} = E^- \cap \auxe^-$,  $E^\texttt{-+} = E^- \cap \auxe^+$,  $E^\texttt{+-} = E^+ \cap \auxe^-$, and $E^\texttt{++} = E^+ \cap \auxe^+$.
We further refine $E^\texttt{+-}$ into two different edge classes:
\begin{align*}
	E^\texttt{b} &= \{ij \in E^\texttt{+-}\colon ij \in E_\mathcal{D} \text{ and } \rho(ij) \in E^\texttt{+-}\}\\
	E^\texttt{o} &= E^\texttt{+-} \setminus E^\texttt{b}.
\end{align*}
We refer to $E^\texttt{b}$ as \emph{balanced} edges, since each edge $ij \in E^\texttt{b}$ is balanced by a partner $\rho(ij) \in E^\texttt{b}$ that is also flipped from positive to negative. Set $E^\texttt{o}$ captures all \emph{other} edges from $E^\texttt{+-}$ (see Figure~\ref{fig:intermediate}). In total, this partitions edges into five different classes, which we denote by $E^\texttt{cl}$ for $\texttt{cl} \in \{\texttt{-}\texttt{-},\texttt{-+},\texttt{++},\texttt{b},\texttt{o}\}$. See Figure~\ref{fig:badtri}.

\begin{lemma}
	\label{lem:edgefacts}
	For two nodes $\{i,j\}$ where $S(i) \neq S(j)$,
	\begin{enumerate}[label = (\alph*)]
		\item $ij \in E_\mathcal{D} \implies x_{ij} + x_{\rho(ij)} \geq 1$. 
		\item $ij \in E^\texttt{++} \implies x_{ij} < \frac23$.
		\item $ij \in E^\texttt{-+}\cup E^\texttt{o} \implies x_{ij} \geq \frac{1}{3}$.
		\item $|E^\texttt{b}| \leq 2\sum_{ij \in E^\texttt{b}} x_{ij}$.
	\end{enumerate}
\end{lemma}
\begin{proof}
	We provide a short proof for each property.
	
\textbf{Property (a).} Let $k\ell = \rho(ij)$, where nodes $k$ and $\ell$ satisfy $S(j) = S(k)$ and $(S(i), S(\ell)) \in \hostile$. Since $S(i)$ and $S(\ell)$ are hostile, we know $X_{i\ell}^- = 0$ by an LP constraint. Since $ij$ and $\rho(ij)$ are positive edges, we have $x_{ij} = X_{ij}^+$ and $x_{\rho(ij)} = x_{k\ell} = X_{k\ell}^+ = X_{j\ell}^+$. From the LP constraints we see $x_{ij} + x_{\rho(ij)} + 0 = X_{ij}^+ + X_{j\ell}^+ + X_{i\ell}^- \geq 1$.
	
	\textbf{Property (b).} Note that $ij \in E^+ \implies x_{ij} = X_{ij}^+$, and the construction of $\auxgraph$ implies $X_{ij}^+ < 2/3$.
	
	\textbf{Property (c).} If $ij \in E^\texttt{-+}$, then $x_{ij} = X_{ij}^-$ and we know $X_{ij}^+ < 2/3$ or else $ij$ would be negative in $\auxgraph$. Thus, the constraint $X_{ij}^+ + X_{ij}^- \geq 1$ implies $x_{ij} = X_{ij}^- \geq 1 - X_{ij}^+ > 1 - 2/3 = 1/3$. 
	
	For $ij \in E^\texttt{o}$, we know that $x_{ij} = X_{ij}^+$, and we consider several different cases to prove $x_{ij} \geq 1/3$. First of all, if $X_{ij}^+ \geq X_{ij}^-$, then constraint $X_{ij}^+ + X_{ij}^- \geq 1$ implies that $x_{ij} = X_{ij}^+ \geq 1/2$. Assume then that $X_{ij}^- > X_{ij}^+$, and break this into two subcases based on whether or not $ij \in E_\mathcal{D}$. If $X_{ij}^- > X_{ij}^+$ and $ij \notin E_\mathcal{D}$, we know that the sign for edges between $I = S(i)$ and $J = S(j)$ in $\auxgraph$ is determined by the LP values. 
	From the construction of $\auxgraph$, we know that edges between $I$ and $J$ are negative in $\auxgraph$ either because (i) $X_{IJ}^+ \geq X_{IJ}^-$ or (ii) $X_{IJ}^- > X_{IJ}^+$ and $X_{IJ}^+ \geq 2/3$. Because we assumed $X_{IJ}^- > X_{IJ}^+$, we know $x_{ij} = X_{IJ}^+ \geq 2/3$. The last case to consider for $ij \in E^\texttt{o}$ is when $X_{ij}^- > X_{ij}^+$ and $ij \in E_\mathcal{D}$. Because $ij \notin E^\texttt{b}$, we know that $\rho(ij) \in E^\texttt{++}$. By Property (b) we have $x_{\rho(ij)} < 2/3$, and Property (a) then implies that $x_{ij} \geq 1 - x_{\rho(ij)} \geq 1 - 2/3 = 1/3$. 
	
	\textbf{Property (d).} Recall that $E^\texttt{b}$ is defined so that $ij \in E^\texttt{b} \implies \rho(ij) \in E^\texttt{b}$, and from Property (a) we know that $x_{ij} + x_{\rho(ij)} \geq 1$. Since edges in $E^\texttt{b}$ come in pairs, and since the sum of LP values for each pair is at least 1, we know that $|E^\texttt{b}|/2 \leq\sum_{ij \in E^\texttt{b}}x_{ij}$.
\end{proof}

\begin{figure}[t]
	\centering
	\includegraphics[width=0.6\linewidth]{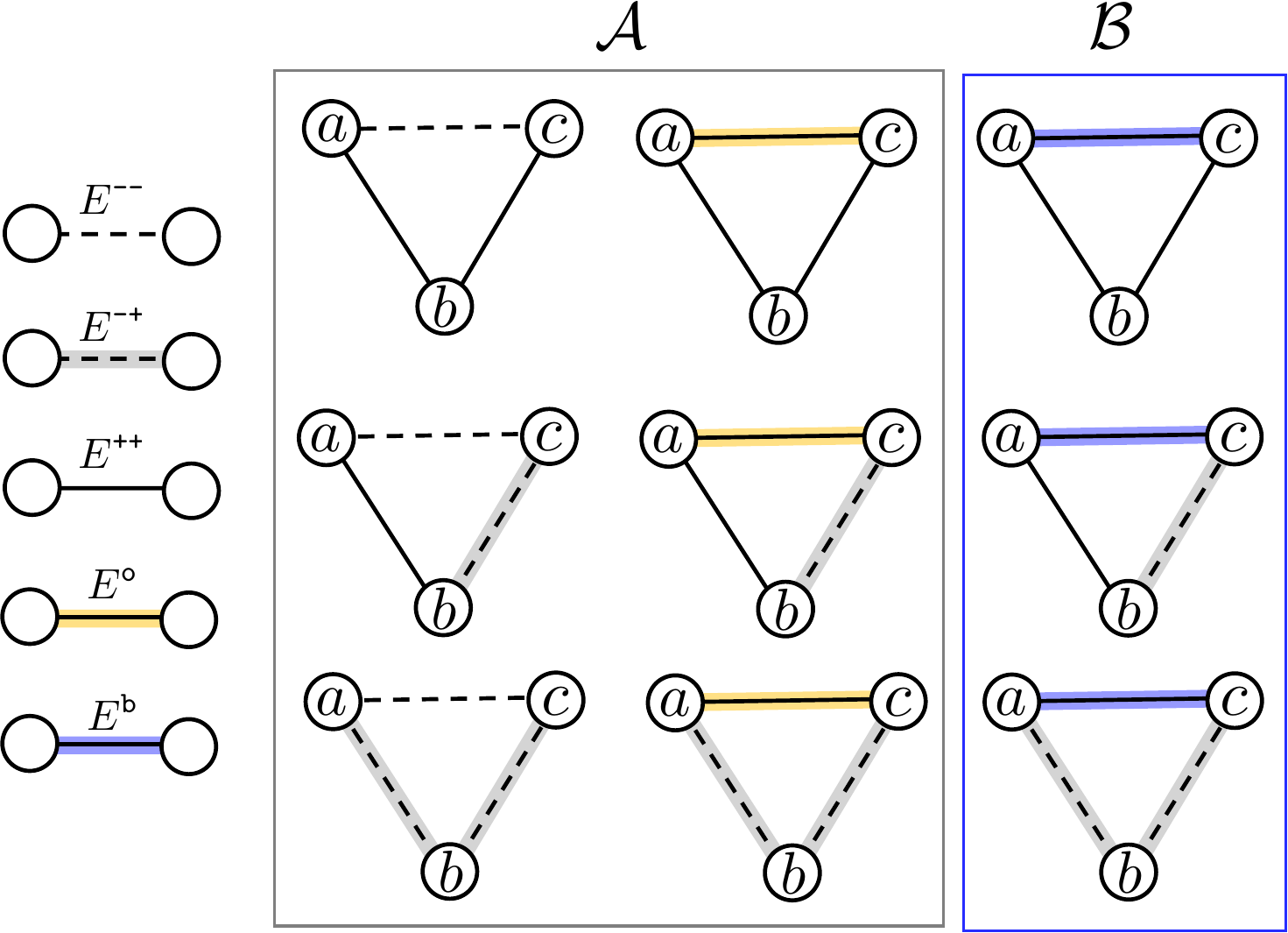}
	\caption{We illustrate the 5 edge classes in $G$, and resulting types of bad triangles in $\auxgraph$.} 
	\label{fig:badtri}
\end{figure}
Separating edges into classes leads to several different types of bad triangles in $\auxgraph$. Recall that $abc \in \badtri(\auxgraph)$ means that $\{ab, bc\} \subseteq\auxe^+$ and $ac \in \auxe^-$. There are two possibilities for $ab$: either $ab \in E^\texttt{++}$ or $ab \in E^\texttt{-+}$. Similarly, either $bc \in E^\texttt{++}$ or $bc \in E^\texttt{-+}$. There are three possible edge classes $\{ E^{\texttt{-}\texttt{-}}, E^\texttt{b},E^\texttt{o}\}$ for $ac \in \auxe^-$.  There are therefore 9 different types of bad triangles in $\auxgraph$, based on the edge type for $ac$ (3 choices), and how many of $\{ab, bc\}$ are in $E^\texttt{-+}$ (3 choices: 0, 1, or 2). See Figure~\ref{fig:badtri}. For our analysis, we also distinguish between bad triangles in $\auxgraph$ based on whether they involve a balanced edge (type-$\mathcal{B}$ bad triangles) or not (type-$\mathcal{A}$ bad triangles). Formally, define
\begin{align}
	\label{eq:typicaltri}
	\mathcal{A} &= \{ abc \colon ab \in \auxe^{+}, bc \in \auxe^{+},  ac \in \auxe^{-} \setminus E^\texttt{b}\} &(\text{first 2 columns of triangles in Figure~\ref{fig:badtri}})\\
	\label{eq:bluetri}
	\mathcal{B} &= \{ abc \colon ab \in \auxe^{+}, bc \in \auxe^{+},  ac \in E^\texttt{b}\}  &(\text{last column of triangles in Figure~\ref{fig:badtri}}).
\end{align}
\begin{lemma}
	\label{lem:badtri}
	If $abc \in \badtri(\auxgraph)$, then $\{S(a),S(b),S(c)\}$ are all distinct. Furthermore, if $abc \in \mathcal{A}$, then
	\begin{equation}
		\label{eq:triinequality}
		\mathbbm{1}(ab \in E^+) + \mathbbm{1}(bc \in E^+) +	\mathbbm{1}(ac \in E^-) \leq 3 \cdot (x_{ab} + x_{bc} + x_{ac}).
	\end{equation}
\end{lemma}
\begin{proof}
	If any of $\{a,b,c\}$ shared the same supernode, one can easily show this would violate the fact that $\auxgraph$ is pivot-safe.
	The left hand side of inequality~\eqref{eq:triinequality} is the number of edges in $abc$ that are \textit{not} flipped when building $\auxgraph$.
	The proof follows from a cases analysis on edge $ac \in E^{\texttt{-}\texttt{-}} \cup E^\texttt{o}$.
	
	\textbf{Case 1: $ac \in E^{\texttt{-}\texttt{-}}$.} In this case, $x_{ac} = X_{ac}^-$. Furthermore, we know that $x_{ab} \geq  X_{ab}^+$ because $ab \in E^\texttt{++} \implies x_{ab} = X_{ab}^+$ whereas $ab \in E^\texttt{-+} \implies x_{ab} = X_{ab}^- > X_{ab}^+$ by the construction of $\auxgraph$. By the same argument, we know $x_{bc} \geq X_{bc}^+$. Combining this with the LP constraints gives:
	\begin{equation}
		\label{eq:seq}
		\mathbbm{1}(ab \in E^+) + \mathbbm{1}(bc \in E^+) +	\mathbbm{1}(ac \in E^-) \leq 3 \leq 3\cdot (X_{ab}^+ + X_{bc}^+ + X_{ac}^-) \leq 3\cdot (x_{ab} + x_{bc} + x_{ac}).
	\end{equation}
	
	\textbf{Case 2: $ac \in E^\texttt{o}$ and $X_{ac}^+ \geq X_{ac}^-$.}
	In this case $x_{ac} = X_{ac}^+ \geq X_{ac}^-$, and the inequalities in~\eqref{eq:seq} still apply.

	\textbf{Case 3: $ac \in E^\texttt{o}$,  $X_{ac}^- > X_{ac}^+$, and $ac \notin E_\mathcal{D}$.} Because $ac \notin E_\mathcal{D}$, the construction of $\auxgraph$ implies that the sign of edges between $S(a)$ and $S(c)$ is determined by LP values. The only way for $ac$ to be negative in $\auxgraph$ while $X_{ac}^- > X_{ac}^+$  is if $x_{ac} = X_{ac}^+ \geq 2/3$. Since there are at most two unflipped edges in the triangle, we have
	\begin{align*}
		\mathbbm{1}(ab \in E^+) + \mathbbm{1}(bc \in E^+) +	\mathbbm{1}(ac \in E^-) \leq 2 \leq 3 \cdot (x_{ac}) \leq 3\cdot (x_{ab} + x_{bc} + x_{ac}).
	\end{align*}
	
	\textbf{Case 4: $ac \in E^\texttt{o}$,  $X_{ac}^- > X_{ac}^+$, and $ac \in E_\mathcal{D}$.} If $\{ab, bc\} \subseteq E^\texttt{-+}$, then $\mathbbm{1}(ab \in E^+) + \mathbbm{1}(bc \in E^+) + \mathbbm{1}(ac \in E^-) = 0$ so the inequality holds trivially. If only one of these edges is in $E^\texttt{-+}$, then nodes $\{a,b,c\}$ define a bad triangle in $G$ so LP constraints guarantee $x_{ab} + x_{bc} + x_{ac} \geq 1$, hence
	\begin{align*}
		\mathbbm{1}(ab \in E^+) + \mathbbm{1}(bc \in E^+) +	\mathbbm{1}(ac \in E^-) = 1 \leq (x_{ab} + x_{bc} + x_{ac}).
	\end{align*}
	Finally, if $ab$ and $bc$ are both in $E^\texttt{++}$, then $abc$ is a triangle of all positive edges in $G$. Because $ac \in E_\mathcal{D}$, we know that the LP includes a \heap{} constraint $x_{\rho(ac)} + x_{ab} + x_{bc} \geq 1$ (see Algorithm~\ref{alg:heap}). Furthermore, $\rho(ac) \in E^\texttt{++}$ or else $ac$ would be in $E^\texttt{b}$ instead of $E^\texttt{o}$. Property (b) of Lemma~\ref{lem:edgefacts} shows $x_{\rho(ac)} < \frac23$. Therefore, $x_{ab} + x_{bc} \geq 1 - x_{\rho(ac)} > 1- \frac23 = \frac13$. Combining this with $x_{ac} + x_{\rho(ac)} \geq 1$ (Property (a) of Lemma~\ref{lem:edgefacts}) gives
	\begin{align*}
		x_{ab} + x_{bc} + x_{ac} \geq x_{ab} + x_{bc} + 1 - x_{\rho(ac)} \geq \frac23 = \frac13 \left(\mathbbm{1}(ab \in E^+) + \mathbbm{1}(bc \in E^+) +	\mathbbm{1}(ac \in E^-) \right).
	\end{align*}
\end{proof}

\subsection{Approximation analysis for \textsf{Pivot}}
We now consider applying \textsf{Pivot} to $\auxgraph$ (see Algorithm~\ref{alg:pivot}). Recall that $\auxgraph_i = (V_i, \auxe_i^+, \auxe_i^-)$ is the subgraph right before choosing the $i$th pivot node.
\begin{figure}[t]
	\centering
	\includegraphics[width=0.9\linewidth]{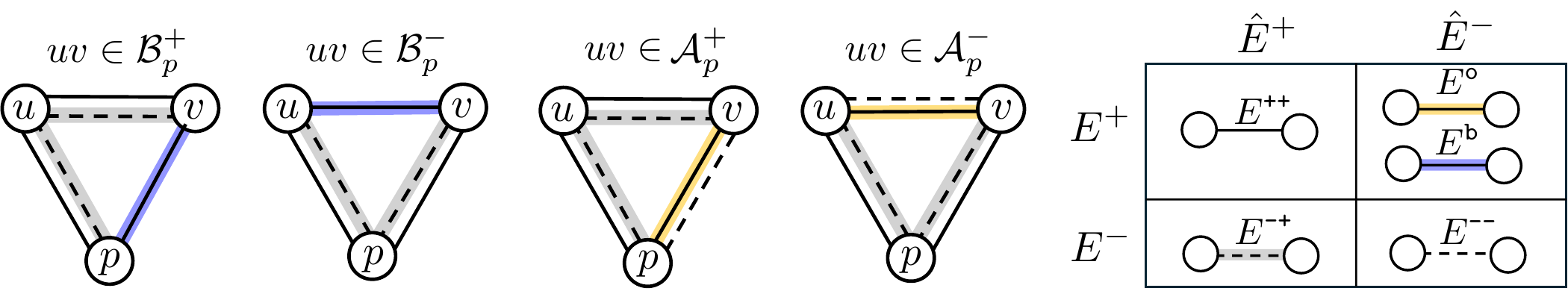}
	\caption{If an edge $uv$ is clustered in a way that disagrees with its edge sign in $\auxgraph$, this is always due to a bad triangle $\{u,v,p\}$ in $\auxgraph$ where $p$ is the pivot node. This figure highlights four different cases for such a bad triangle $\{u,v,p\}$. The first two correspond to triangles from $\mathcal{B}$, and the second two correspond to triangles from $\mathcal{A}$; see Figure~\ref{fig:badtri}. For each edge $uv$, we indicate all of the different classes $uv$ could correspond to in each case.}
	\label{fig:puvtypes}
\end{figure}
If $p \in V_i$ is chosen as the $i$th pivot node, this would form a cluster $C_p = \{p\} \cup \{v \in V_i \colon vp \in \auxe_i^+\}$. An edge $uv$ would then be clustered in a way that disagrees with its sign in $\auxgraph_i$ if and only if $\{u,v,p\}$ is a bad triangle in $\gi{}$. 
This edge $uv$ would belong to one of four sets, based on whether $uv \in \auxe^+_i$ or $uv \in \auxe^-_i$, and whether the bad triangle $\{u,v,p\}$ is of type-$\mathcal{A}$ or type-$\mathcal{B}$ (see Figure~\ref{fig:puvtypes}):
\begin{align*}
	\mathcal{B}_p^+(\gi{}) &= \{uv \in \auxe_i^+ \colon pu \in \auxe_i^+, pv \in E^\texttt{b}\} & 	\mathcal{A}_p^+(\gi{}) &= \{uv \in \auxe_i^+ \colon pu \in \auxe_i^+, pv \in \auxe_i^- \setminus E^\texttt{b}\} \\
	\mathcal{B}_p^-(\gi{}) &= \{uv \in E^\texttt{b} \colon pu \in \auxe_i^+, pv \in \auxe_i^+\} &
	\mathcal{A}_p^-(\gi{}) &= \{uv \in \auxe_i^- \setminus E^\texttt{b} \colon pu \in \auxe_i^+, pv \in \auxe_i^+\}.
\end{align*}
To match this to edge sets in Lemma~\ref{lem:3pt1}, note that $\badtri_p^+(\gi{}) = \mathcal{A}_p^+(\gi{}) \cup \mathcal{B}_p^+(\gi{})$ and $\badtri_p^-(\gi{}) = \mathcal{A}_p^-(\gi{}) \cup \mathcal{B}_p^-(\gi{})$. In what follows, we omit $\auxgraph_i$ when it clear from context, e.g., $\mathcal{B}_p^+ = \mathcal{B}_p^+(\gi{})$.

An edge $uv \in \badtri_p^+ \cup \badtri_p^-$ is only a mistake for the \ccc{} instance $G$ if $uv \in E^{\texttt{-}\texttt{-}} \cup E^\texttt{++}$, i.e., $uv$ is not flipped.
When pivoting on $p$, the number of mistakes made at unflipped edges is equal to $|\badtri_p^+ \cap E^+| + |\badtri_p^- \cap E^-|$. The following lemma bounds these types of mistakes.
\begin{lemma}
	\label{lem:pivotexists}
	There exists some $p \in V_i$ satisfying
	\begin{align}
		\label{eq:bound}
		|\badtri_p^+ \cap E^+| + |\badtri_p^-  \cap E^-|  - 2|\mathcal{B}_p^-| - \sum_{uv \in \badtri^+_p \cup \mathcal{A}_p^-} 3x_{uv} \leq 0.
	\end{align}
	If a pivot node is chosen uniformly at random, this bound holds in expectation.
\end{lemma}
\begin{proof}
	Let $\mathcal{B}_i \subseteq \mathcal{B}$ denote type-$\mathcal{B}$ bad triangles in $\gi$. Define $\mathcal{A}_i \subseteq \mathcal{A}$ analogously. Because a bad triangle has two positive edges and one negative edge, summing up the values of $|\mathcal{B}_p^+|$ across all choices of $p$ is equivalent to counting each type-$\mathcal{B}$ bad triangle in $\gi$ twice. Similarly, summing up $|\mathcal{B}_p^-|$ across all choices of $p$ amounts to counting each type-$\mathcal{B}$ triangle once. Therefore,
	\begin{align*}
		\sum_{p \in V_i} |\mathcal{B}_p^+ \cap E^+| \leq \sum_{p \in V_i} |\mathcal{B}_p^+| = 2 |\mathcal{B}_i| = 2 \sum_{p \in V} |\mathcal{B}_p^-|.
	\end{align*}
	Using similar logic, and applying the bound in Lemma~\ref{lem:badtri} for type-$\mathcal{A}$ bad triangles, we get
	\begin{align*}
		\sum_{p \in V_i} |\mathcal{A}_p^+ \cap E^+| + |\mathcal{A}_p^- \cap E^-| &= \sum_{abc \in \mathcal{A}_i} \mathbbm{1}(ab \in E^+) + \mathbbm{1}(bc \in E^+) +	\mathbbm{1}(ac \in E^-)\\
		&\leq \sum_{abc \in \mathcal{A}_i} 3(x_{ab} + x_{bc} + x_{ac}) = \sum_{p \in V_i} \sum_{uv \in \mathcal{A}_p^+ \cup \mathcal{A}_p^-} 3x_{uv} \leq \sum_{p \in V_i} \sum_{uv \in \badtri_p^+ \cup \mathcal{A}_p^-} 3x_{uv}.
	\end{align*}
	Combining these inequalities with the observation that $\badtri_p^- \cap E^- = \mathcal{A}_p^- \cap E^-$ and $\badtri_p^+ = \mathcal{A}_p^+ \cup \mathcal{B}_p^+$ gives
	\begin{align*}
	\sum_{p \in V_i}  |\badtri_p^+ \cap E^+| + |\badtri_p^- \cap E^-| =	\sum_{p \in V_i} |\mathcal{B}_p^+ \cap E^+| + |\mathcal{A}_p^+ \cap E^+| + |\mathcal{A}_p^- \cap E^-| \leq \sum_{p\in V_i} \Big( 2 |\mathcal{B}_p^-| + \sum_{uv \in \badtri_p^+ \cup \mathcal{A}_p^-} 3x_{uv}\Big).
	\end{align*}
	By considering all possible pivots $p$, we must find at least one pivot satisfying inequality~\eqref{eq:bound}. If $p$ is chosen uniformly at random, the left side of inequality~\eqref{eq:bound} is a random variable with expected value at most 0.
\end{proof}
In order to define a deterministic pivoting strategy that fits the template in Lemma~\ref{lem:3pt1}, we define
\begin{equation}
	\label{yuvfinal}
	y_{uv} = \begin{cases}
		2 & \text{if $uv \in E^\texttt{b}$} \\
		3x_{uv} & \text{otherwise},
	\end{cases}
\end{equation}
where $\{x_{uv}\}_{uv \in E^+ \cup E^-}$ is a $(1+\varepsilon/3)$-approximate solution for LP~\eqref{eq:ccc_lp}.
\begin{theorem}
	\label{thm:main}
	Using the edge budgets in Eq.~\eqref{yuvfinal}, Algorithm~\ref{alg:mainalg} returns a $(3+\varepsilon)$-approximate solution for \ccc{} if in iteration $i$ of running \textsf{Pivot} on $\auxgraph$, a pivot node $p$ is chosen to minimize the ratio
	\begin{equation}
		\label{eq:ratio}
	\frac{|\badtri^+_p(\auxgraph_i) \cap E^+| + |\badtri^-_p(\auxgraph_i) \cap E^-| }{ \sum_{uv \in \badtri^+_p(\auxgraph_i) \cup \badtri^-_p(\auxgraph_i)} y_{uv}}.
	\end{equation}
For uniform random pivot nodes, Algorithm~\ref{alg:mainalg} returns a $(3+\varepsilon)$-approximation in expectation.
\end{theorem}
\begin{proof}
	We first prove the result for deterministic pivot choices. Let $\pivset \subseteq V$ represent the sequence of pivot nodes obtained by minimizing the ratio in~\eqref{eq:ratio} at each iteration. Minimizing this ratio guarantees we satisfy inequality~\eqref{eq:bound} in Lemma~\ref{lem:pivotexists} at each step. Let $\mathcal{C} = \{C_p \colon p \in \pivset\}$ be the resulting clustering. For each edge class $\textsf{cl} \in \{\texttt{-}\texttt{-}, \texttt{-+}, \texttt{-+}, \texttt{b}, \texttt{o}\}$, let 
	\begin{align*}
		B^{\textsf{cl}} = \{uv \in E^{\textsf{cl}} \colon \mathcal{C}(u) \neq \mathcal{C}(v)\} \quad \text{ and } \quad I^{\textsf{cl}} = \{uv \in E^{\textsf{cl}} \colon \mathcal{C}(u) = \mathcal{C}(v)\}
	\end{align*}
	denote edges from $E^\texttt{cl}$ on the boundary and interior of $\mathcal{C}$ respectively.
	The cost of $\mathcal{C}$ is then given by
	\begin{align}
		\label{eq:cost}
		\text{cost}_G(\mathcal{C}) &= |\imm| + |\imp| + |\bpp| +|\bb| + |\bo|.
	\end{align}

	For each $p \in \mathcal{P}$, let $\badtri_p^+ = \badtri_p^+(\auxgraph_{\iter{p}})$ where $\iter{p}$ is the iteration in which node $p$ was chosen as pivot, and define $\badtri_p^-$, $\mathcal{A}_p^+$, $\mathcal{A}_p^-$, $\mathcal{B}_p^+$, and $\mathcal{B}_p^-$ analogously.
	Then
	\begin{equation}
		\label{bppimm}
		\sum_{p \in \mathcal{P}} |\badtri^+_p \cap E^+| + |\badtri^-_p \cap E^-| = |\bpp| + |\imm|.
	\end{equation} 
	Similarly, $\mathcal{B}_p^-$ denotes edges from $E^\texttt{b}$ that ended up inside of $C_p$ when $p$ was chosen as a pivot, so
	\begin{equation}
		\label{ib}
		\sum_{p \in\mathcal{P}} 2 |\mathcal{B}_p^-| = 2 |I^\texttt{b}|.
	\end{equation}
	Next, observe that $\badtri_p^+$ captures all edges from $E^\texttt{++} \cup E^{\texttt{-+}}$ that end up on the boundary of $C_p$ when $p$ is chosen as pivot, and $\mathcal{A}_p^-$ captures all edges from $E^{\texttt{-}\texttt{-}} \cup E^\texttt{o}$ that end up inside $C_p$ (see Figure~\ref{fig:puvtypes}). This means that
	\begin{equation}
		\label{wp}
		\bigcup_{p \in \mathcal{P}} \badtri_p^+ \cup \mathcal{A}_p^- = B^\texttt{++} \cup B^{\texttt{-+}} \cup I^{\texttt{-}\texttt{-}} \cup I^\texttt{o}.
	\end{equation}
	Define $W_\mathcal{P} = B^\texttt{++} \cup B^{\texttt{-+}} \cup I^{\texttt{-}\texttt{-}} \cup I^\texttt{o}$. Combining Eq.~\eqref{bppimm}-\eqref{wp} with inequality~\eqref{eq:bound} gives
	\begin{equation}
		\label{eq:hardest}
		|\bpp| + |\imm| = \sum_{p \in \pivset} |\badtri^+_p \cap E^+| + |\badtri^-_p \cap E^-|
		\leq \sum_{p \in \pivset} \Big( 2|\mathcal{B}_p^-| + \sum_{uv \in \badtri_p^+ \cup \mathcal{A}_p^-} 3x_{uv} \Big)
		= 2|I^\texttt{b}| + \sum_{uv \in W_\mathcal{P}} 3 x_{uv}.
	\end{equation}
	If $uv \in I^\texttt{b}$, then $\rho(uv) \in B^\texttt{b}$, or there would be a constraint violation. Thus, $|\ib| \leq |E^\texttt{b}|/2$. Furthermore, Property (d) of Lemma~\ref{lem:edgefacts} shows $|E^\texttt{b}| \leq \sum_{uv \in E^\texttt{b}} 2x_{uv}$.  Combining this with inequality~\eqref{eq:hardest} gives
	\begin{equation}
		\label{eq:hardblue}
		|\bb| + |\bpp| + |\imm| \leq |\bb| + 2|I^\texttt{b}| + \sum_{uv \in W_\mathcal{P}} 3 x_{uv} = |E^\texttt{b}| + |I^\texttt{b}| + \sum_{uv \in W_\mathcal{P}} 3 x_{uv} 
		\leq \sum_{uv \in W_\mathcal{P}\cup E^\texttt{b}} 3 x_{uv}.
	\end{equation}
	If $uv \in E^\texttt{o} \cup E^\texttt{-+}$, then $x_{uv} \geq \frac{1}{3}$ by Property (c) of Lemma~\ref{lem:edgefacts}. Thus
	\begin{equation}
		\label{eq:thirdeasy}
		|\bo| + |\imp| \leq \sum_{uv \in B^\texttt{o} \cup \imp} 3x_{uv}.
	\end{equation}
	Because $(\bo \cup \imp)$, $E^\texttt{b}$, and $W_\mathcal{P}$ are all disjoint sets of edges, the inequalities in~\eqref{eq:hardblue} and~\eqref{eq:thirdeasy} imply that
	\begin{equation}
		\text{cost}_G(\mathcal{C}) = 	|\bb| + |\bpp| + |\imm| +|\bo| + |\imp| \leq \sum_{uv \in E^+ \cup E^-} 3 x_{uv}.
	\end{equation}
	Algorithm~\ref{alg:mainalg} assumes $\{x_{uv}\}_{uv \in E^+ \cup E^-}$ is a $(1+ \varepsilon/3)$-approximate solution to the LP relaxation, so $\text{cost}_G(\mathcal{C})$ is within a factor $(3+\varepsilon)$ of the optimal \ccc{} solution for our deterministic pivoting strategy.
	
	If pivot nodes are chosen uniformly at random, then Lemma~\ref{lem:pivotexists} guarantees inequality~\eqref{eq:hardest} holds in expectation. Furthermore, the inequality $|E^\texttt{b}| \leq 2 \sum_{uv \in E^\texttt{b}} x_{uv}$ and inequality~\eqref{eq:thirdeasy} hold for every choice of pivot nodes. Thus, the $(3+\varepsilon)$-approximation holds in expectation for uniform random pivot nodes.
\end{proof}
For a runtime analysis, recall that $\mathcal{D}$ and $\mathcal{H}$ can be computed in $O(n^3)$ time via Algorithms~\ref{alg:computeD} and~\ref{alg:heap}. Theorem~\ref{thm:covering} guarantees that $\mathcal{X}$ can be computed in $\tilde{O}(n^3)$ time. Given $\mathcal{X}$, building $\auxgraph$ takes $O(n^2)$ time since it just requires iterating through each edge once. Running \textsf{Pivot} in an $n$-node graph with uniform random pivot nodes takes $O(n^2)$ time, and Lemma~\ref{lem:3pt1} guarantees that our deterministic pivoting strategy takes $O(n^3)$ time. 

\section{\ccc{} with One-sided Constraints}
We now consider special cases of \ccc{} with one-sided constraints, meaning that there are only friendly constraints ($\hostile = \emptyset$), or only hostile constraints ($\friendly = \emptyset$). We call the former \fcc{} and the latter \hcc{}. For these problems, we can design algorithms that are even simpler (both in terms of their presentation and their analysis), while maintaining the same (or slightly better) guarantees. For \fcc{}, we design a simpler rounding algorithm for a covering LP that does not involve any \heap{} constraints. For \hcc{}, we can completely avoid solving a covering LP and still obtain a (randomized) 3-approximation in $O(n^3)$ time. 

\subsection{Only friendly constraints}
Let $G = (V,E^+,E^-,\friendly,\emptyset{})$ be an instance of \fcc{}. As before, we have a collection of supernodes $\mathcal{S}$ defined by friendly constraints, but no pair of supernodes is hostile. We assume $G$ is in consistent form, meaning that $uv \in E^+$ if $s(u) = s(v)$. If this does not hold we can run a simple pre-processing step to satisfy this property without affecting runtime or approximation factor (see Lemma~\ref{lem:consistent}). 
For this case, we consider a simpler covering LP without \heap{} constraints:
\begin{equation}
	\begin{aligned}
		\label{eq:fcc_lp}
		\min & \displaystyle\sum_{uv \in E^+ \cup E^-} x_{uv}  &\\
		\text{s.t. } & x_{uv} = X_{S(u)S(v)}^+ = X_{uv}^+& \text{ $\forall  uv \in E^+$}\\
		& x_{uv} = X_{S(u)S(v)}^- = X_{uv}^-& \text{ $\forall uv \in E^-$}\\
		& X_{AA}^+ = 1-X_{AA}^- = 0 & \text{ $\forall  A \in \mathcal{S}$} \\
		&\begin{rcases}
			X_{AB}^+ + X_{BC}^+ + X_{AC}^- \geq 1 \\ 
			X_{AB}^+ + X_{BC}^- + X_{AC}^+ \geq 1 \\
			X_{AB}^- + X_{BC}^+ + X_{AC}^+ \geq 1 \\
		\end{rcases} & \text{ $\forall  \{A,B,C\} \in {\mathcal{S} \choose 3}$} \\
		& X_{AB}^+ + X_{AB}^- \geq 1 & \text{$\forall  A,B \in \mathcal{S}$}\\
		& X_{AB}^+ \geq 0, X_{AB}^- \geq 0 & \text{$\forall  A,B \in \mathcal{S}$.}	
	\end{aligned}
\end{equation}

\begin{algorithm}[t]
	\caption{$(3+\varepsilon)$-approximation for \fcc{}}
	\label{alg:friendly}
	\begin{algorithmic}[1]
		\State \textbf{Input}: \fcc{} instance $G$ in consistent form, approximation parameter $\varepsilon \in (0,1)$
		\State \textbf{Output}: feasible clustering for \fcc{}
		\State $\mathcal{X} \longleftarrow $ $\big(1+\frac{\varepsilon}{3}\big)$-approximate solution to LP~\eqref{eq:fcc_lp}
		\State $\auxe^+ \longleftarrow  \{uv \colon X_{uv}^- \geq X_{uv}^+\}$
		\State $\auxe^- \longleftarrow  \{uv \colon X_{uv}^+ > X_{uv}^-\}$
		\State $\mathcal{C} \longleftarrow \textsf{Pivot}(\auxgraph = (V,\auxe^+, \auxe^-))$
		\State Return $\mathcal{C}$
	\end{algorithmic}
\end{algorithm}
Algorithm~\ref{alg:friendly} is pseudocode for rounding an approximate solution to this LP relaxation. It constructs a new graph $\auxgraph = (V, \auxe^+, \auxe^-)$, where the sign of edges between two supernodes $A$ and $B$ is determined only by whether $X_{AB}^- \geq X_{AB}^+$. The construction guarantees that every edge $uv$ in the same supernode defines a positive edge, and furthermore that $u$ and $v$ have the same positive and negative neighbors. Hence, applying \textsf{Pivot} to $\auxgraph$ satisfies all friendly constraints. As before, let $V_i$ denote nodes that remain unclustered after $i -1$ pivot steps, $\auxe^-_i = \auxe^- \cap (V_i \times V_i)$, and $\auxe^+_i = \auxe^+ \cap (V_i \times V_i)$. For arbitrary $p \in V_i$, define
\begin{align}
	\label{eq:mpp}
	\badtri^+_p(\gi) = \{uv \in \auxe_i^+ \colon pu \in \auxe_i^+, pv \in \auxe_i^-\} \text{ and } 
	\badtri^-_p(\gi) = \{uv \in \auxe_i^- \colon pu \in \auxe_i^+, pv \in \auxe_i^+\}.
\end{align}
\begin{theorem}
	\label{thm:fcc}
	For constant $\varepsilon > 0$, Algorithm~\ref{alg:friendly} is a $(3+\varepsilon)$-approximation for \fcc{} if the $i$th pivot node is chosen to minimize
	\begin{align}
		\label{eq:ratiofcc}
		\frac{|\badtri^+_p(\gi) \cap E^+| + |\badtri^-_p(\gi) \cap E^-|}{ \sum_{uv \in \badtri^+_p(\gi) \cup \badtri^-_p(\gi)}x_{uv}}.
	\end{align}
\end{theorem}
\begin{proof}
Consider an arbitrary bad triangle $abc \in \badtri(\auxgraph)$. Let $A = S(a)$, $B = S(b)$, and $C = S(c)$. 
If $ab \in E^+$, then $x_{ab} = X_{ab}^+$, otherwise $x_{ab} = X_{ab}^- \geq X_{ab}^+$. Either way, $x_{ab} \geq X_{ab}^+$. Using a similar argument, we can show that $x_{bc} \geq X_{bc}^+$ and $x_{ac} \geq X_{ac}^-$. Thus, 
\begin{align}
	\label{eq:triplef}
	x_{ab} + x_{bc} + x_{ac} \geq X_{ab}^+ + X_{bc}^+ + X_{ac}^- \geq 1.
\end{align}
For a fixed iteration $i$, let $\badtri_i = \badtri(\gi)$. For an arbitrary $p \in V_i$ write $\badtri_p^+ = \badtri^+_p(\gi)$ and $\badtri_p^- = \badtri^-_p(\gi)$ for simplicity. Using inequality~\eqref{eq:triplef} we see 
\begin{align*}
	\sum_{p \in V_i} |\badtri_p^+ \cap E^+| + |\badtri_p^- \cap E^-|  
	& \leq \sum_{p \in V_i} |\badtri_p^+| + |\badtri_p^-| = 3 |\badtri_i|  \leq 3\sum_{abc \in \badtri_i} x_{ab} + x_{bc} + x_{ac}= 3\sum_{p \in V_i} \sum_{uv \in \badtri_p^+ \cup \badtri_p^-} x_{uv}.
\end{align*}
Thus, minimizing the ratio in~\eqref{eq:ratiofcc} guarantees that at iteration $i$ we choose some $p$ satisfying:
\begin{align}
	\label{eq:pivotsat}
	|\badtri_p^+ \cap E^+| + |\badtri_p^- \cap E^-|   \leq \sum_{uv \in \badtri_p^+ \cup \badtri_p^-} 3x_{uv}.
\end{align}

Let $\mathcal{C}$ be the clustering obtained by pivoting in $\auxgraph$ using this deterministic strategy. To bound the total number of mistakes, first partition edges based on their signs in both $G$ and $\auxgraph$:
\begin{align*}
	E^{\texttt{-}\texttt{-}} = E^- \cap \auxe^-, \quad  E^\texttt{-+} = E^- \cap \auxe^+, \quad  E^\texttt{++} = E^+ \cap \auxe^+, \quad  E^\texttt{+-} = E^+ \cap \auxe^-.
\end{align*} 
For edge class $\texttt{cl} \in \{\texttt{++}, \texttt{+-}, {\texttt{-}\texttt{-}}, \texttt{-+}\}$, let $B^\texttt{cl}$ and $I^\texttt{cl}$ denote edges from $E^\texttt{cl}$ that end up on the boundary or interior of clusters respectively. The cost of the clustering $\mathcal{C}$ is then given by
\begin{align*}
	\text{cost}_G(\mathcal{C}) = |B^\texttt{++}| + |B^\texttt{+-}| + |I^{\texttt{-}\texttt{-}}| + |I^\texttt{-+}|.
\end{align*}
Inequality~\eqref{eq:pivotsat} shows that summing across all pivot nodes chosen throughout the algorithm gives
\begin{align*}
	|B^\texttt{++}|  + |I^{\texttt{-}\texttt{-}}| \leq \sum_{uv \in B^\texttt{++} \cup B^\texttt{-+}\cup I^{\texttt{-}\texttt{-}} \cup I^\texttt{+-}} 3 x_{uv}.
\end{align*}
Furthermore, note that $uv \in E^\texttt{+-} \cup E^\texttt{-+} \implies x_{uv} \geq \frac12$. In more detail, if $uv \in E^\texttt{+-}$, then $x_{uv} = X_{uv}^+ > X_{uv}^-$, so the inequality $X_{uv}^+ + X_{uv}^- \geq 1$ implies that $X_{uv}^+ \geq \frac12$. An analogous argument can be shown for $uv \in E^\texttt{-+}$. Thus,
\begin{align}
	|B^\texttt{++}| + |B^\texttt{+-}| + |I^{\texttt{-}\texttt{-}}| + |I^\texttt{-+}| \leq \sum_{uv \in B^\texttt{++} \cup B^\texttt{-+}\cup I^{\texttt{-}\texttt{-}} \cup I^\texttt{+-}}  3 x_{uv} + \sum_{uv \in B^\texttt{+-} \cup I^\texttt{-+}} 2x_{uv} \leq \sum_{uv \in E^+ \cup E^-} 3x_{uv}.
\end{align}
Since Algorithm~\ref{alg:friendly} starts with a $(1+ \varepsilon/3)$-approximate solution for the LP relaxation, this cost is within a factor 3 of the optimal \fcc{} cost.		
\end{proof}
As was the case for Theorem~\ref{thm:main}, the approximation guarantee will hold in expectation if we choose pivot nodes uniformly at random. We can also use similar arguments to prove an $\tilde{O}(n^3)$ runtime guarantee.

\subsection{Only hostile constraints}
Let $G = (V,E^+,E^-,\emptyset{},\hostile)$ be an instance of \hcc{} where $\hostile \subseteq E^-$.

\textbf{Constructing $\auxgraph$.}
We cannot directly apply \textsf{Pivot} to $G$, as this could lead to a hostile edge violation. In particular, $G$ may contain a triplet of nodes $\{a,b,c\}$ such that $ab \in E^+$, $bc \in E^+$, and $ac \in \hostile$. Observe that $(ab,bc)$ defines a dangerous pair in $G$. Because \hcc{} is just an instance of \ccc{} with supernodes of size 1, all dangerous pairs in $G$ correspond to a triangle of this form, with two positive edges and one hostile edge.
To avoid hostile edge violations when pivoting, we first find a maximal edge-disjoint set of dangerous pairs $\mathcal{D}$. Let $E_\mathcal{D}$ represent the set of edges defining $\mathcal{D}$. For $uv \in E_\mathcal{D}$, let $\rho(uv) \in E_\mathcal{D}$ be the edge such that $(uv, \rho(uv)) \in E_\mathcal{D}$. Our approximation algorithm for \hcc{} (Algorithm~\ref{alg:hostile}) forms a graph $\auxgraph$ by flipping all edges in $E_\mathcal{D}$ from positive to negative, and then applies \textsf{Pivot} to $\auxgraph$. 
The maximality of $\mathcal{D}$ ensures that this will produce a feasible clustering for \hcc{} on $G$ using any choice of pivot nodes.

\textbf{The \hcc{} covering LP.} 
To analyze Algorithm~\ref{alg:hostile}, we consider a simple covering LP without supernode variables or \heap{} constraints:
\begin{equation}
	\begin{aligned}
		\label{eq:hcc_lp}
		\min & \sum_{uv \in E^+ \cup E^-} x_{uv}  &\\
		\text{s.t. } & x_{ab} + x_{bc} + x_{ac} \geq 1 & \forall abc \in \badtri(G) \\
		& x_{ab} = 0 & \forall ab \in \hostile \\
		& x_{ab} \geq 0 & \forall ab \notin \hostile.
	\end{aligned}
\end{equation}
Similar to Theorem~\ref{thm:fcc}, we can explicitly solve or approximate this LP, and use the resulting LP variables to design a deterministic pivoting strategy that produces a 3-approximation. Notably, because the construction of $\auxgraph$ does not rely on LP variables, we can also select pivots uniformly at random to obtain a much simpler randomized 3-approximation that runs in $O(n^3)$ time. As before, we let $\gi{} = (V_i, \auxe^+_i, \auxe_i^-)$, represent the subgraph of $\auxgraph$ after selecting the first $i-1$ pivot nodes, $\badtri_i$ be the set of bad triangles in $\gi{}$, and $\badtri_p^+$ and $\badtri_p^-$ be defined as in Eq.~\eqref{eq:mpp}. To define the deterministic pivoting strategy, set a budget for each edge as follows:
\begin{equation}
	\label{yuvhcc}
	y_{uv} = \begin{cases}
		2 & \text{if $uv \in E_\mathcal{D}$}\\
		3x_{uv} & \text{ otherwise},
		\end{cases}
\end{equation}
where $\{x_{uv}\}_{uv \in E^+ \cup E^-}$ is a $(1+\varepsilon/3)$-approximate solution for LP~\eqref{eq:hcc_lp}.
\begin{algorithm}[t]
	\caption{Approximation algorithm for \hcc{}}
	\label{alg:hostile}
	\begin{algorithmic}[1]
		\State \textbf{Input}: \hcc{} instance $G$ where $\hostile \subseteq E^-$
		\State \textbf{Output}: feasible clustering for \hcc{}
		\State $\mathcal{D} \longleftarrow \textsf{Compute}\mathcal{D}(G)$ 
		\State $\auxe^+ \longleftarrow  E^+ \setminus E_\mathcal{D}$
		\State $\auxe^- \longleftarrow  E^- \cup E_\mathcal{D}$
		\State $\mathcal{C} \longleftarrow \textsf{Pivot}(\auxgraph = (V,\auxe^+, \auxe^-))$
		\State Return $\mathcal{C}$
	\end{algorithmic}
\end{algorithm}

\begin{theorem}
	\label{thm:hcc}
	For $\varepsilon > 0$, Algorithm~\ref{alg:hostile} is a deterministic $(3+\varepsilon)$-approximation for \hcc{} if, for the edge budgets in Eq.~\eqref{yuvhcc}, the $i$th pivot node is chosen to minimize the ratio
	\begin{align}
		\label{eq:ratio2}
		\frac{|\badtri^+_p(\gi) \cap E^+| + |\badtri^-_p(\gi) \cap E^-|}{ \sum_{uv \in \badtri^+_p(\gi) \cup \badtri^-_p(\gi)} y_{uv}}.
	\end{align}
	If pivots are chosen uniformly at random, Algorithm~\ref{alg:hostile} is a randomized 3-approximation.
\end{theorem}

\begin{proof}
Define $E^\texttt{++} = E^+ \setminus E_\mathcal{D}$, $E^\texttt{+-} = E_\mathcal{D}$, and $E^{\texttt{-}\texttt{-}} = E^-$ to classify each edge based on its sign in $G$ and $\auxgraph$.
Let $\mathcal{C}$ be a clustering obtained by selecting pivots in $\auxgraph$ to minimize the ratio in Eq.~\eqref{eq:ratio2} at each step. For every class $\texttt{cl} \in \{\texttt{++}, \texttt{+-}, {\texttt{-}\texttt{-}}\}$, let $B^\texttt{cl}$ and $I^\texttt{cl}$ represent the set of edges from $E^\texttt{cl}$ that are between or inside clusters in $\mathcal{C}$ respectively, so that the cost is given by
\begin{align*}
	\text{cost}_G(\mathcal{C}) = |B^\texttt{++}| + |I^{\texttt{-}\texttt{-}}| + |B^\texttt{+-}|.
\end{align*}
To bound this cost, we first separate bad triangles $\badtri_i$ from $\auxgraph_i$ into two classes, based on whether the negative edge is from $E^-$ or $E_\mathcal{D}$:
	\begin{align*}
		\mathcal{A}_i &= \{abc \colon ab \in \auxe_i^+, bc \in \auxe_i^+, ac \in E^-\}\\
		\mathcal{B}_i &=\{abc \colon ab \in \auxe_i^+, bc \in \auxe_i^+, ac \in E_\mathcal{D}\}.
	\end{align*}
	For every $abc \in \mathcal{A}_i$, we know $x_{ab} + x_{bc} + x_{ac} \geq 1$ since $\mathcal{A}_i \subseteq \badtri(G)$. For an arbitrary $p \in V_i$, define
	\begin{align*}
	\mathcal{A}_p^+(\gi{}) &= \{uv \in \auxe_i^+ \colon pu \in \auxe_i^+, pv \in E^-\} & \mathcal{A}_p^-(\gi{}) &= \{uv \in E^- \colon pu \in \auxe_i^+, pv \in \auxe_i^+\}\\
	\mathcal{B}_p^+(\gi{}) &= \{uv \in \auxe_i^+ \colon pu \in \auxe_i^+, pv \in E_\mathcal{D}\} & 
	\mathcal{B}_p^-(\gi{}) &= \{uv \in E_\mathcal{D}\colon pu \in \auxe_i^+, pv \in \auxe_i^+\}.
\end{align*}
	For simplicity we drop the $\gi$ when it is clear from context, e.g.,  $\mathcal{A}_p^+= \mathcal{A}_p^+(\gi)$.

Summing $|\mathcal{A}_p^+| + |\mathcal{A}_p^-|$ across all potential pivot nodes in the graph $\gi$ gives
	\begin{align}
		\label{eq:boundt}
		\sum_{p \in V_i} |\mathcal{A}_p^+| + |\mathcal{A}_p^-| = 3 |\mathcal{A}_i| \leq \sum_{abc \in \mathcal{A}_i} 3(x_{ab} + x_{bc} + x_{ac}) = \sum_{p \in V_i}\sum_{uv \in \mathcal{A}_p^+ \cup \mathcal{A}_p^-} 3x_{uv} \leq \sum_{p \in V_i} \sum_{uv \in \badtri_p^+ \cup \mathcal{A}_p^-} y_{uv}.
	\end{align}
Similarly, we can observe that:
	\begin{align}
		\label{eq:boundb}
		\sum_{p \in V_i} |\mathcal{B}_p^+| = 2 |\mathcal{B}_i| = 2 \sum_{p \in V_i} |\mathcal{B}_p^-| = \sum_{p \in V_i} \sum_{uv \in \mathcal{B}_p^-} y_{uv}.
	\end{align}
	Combining~\eqref{eq:boundt} and~\eqref{eq:boundb}, we know that there is always a pivot node in $\gi$ satisfying
	\begin{align}
		\label{eq:ineqmain}
		|\mathcal{A}_p^+| + |\mathcal{A}_p^-| + |\mathcal{B}_p^+|  \leq  
		\sum_{uv \in \badtri_p^+\cup \badtri_p^-} y_{uv}=
		2 |\mathcal{B}_p^-| + \sum_{uv \in \badtri_p^+ \cup \mathcal{A}_p^-} 3x_{uv}. 
	\end{align}

Observe now that
	$\mathcal{A}_p^+ \cup \mathcal{B}_p^+ = \badtri_p^+ = \badtri_p^+ \cap E^+\subseteq E^\texttt{++}$ is the set of new positive mistakes from $E^\texttt{++}$ that we make by pivoting on $p$, and $\mathcal{A}_p^- = \badtri_p^- \cap E^-$ is the set of negative edges from $G$ that end up inside $C_p$. The set $\mathcal{B}_p^-$ is all edges from $E^\texttt{+-} = E_\mathcal{D}$ that end up inside clusters when pivoting on $p$. From this we can see that the pivot minimizing the ratio in Eq.~\eqref{eq:ratio2} is guaranteed to satisfy inequality~\eqref{eq:ineqmain}, since the left hand side of~\eqref{eq:ineqmain} is equivalent to the numerator of~\eqref{eq:ratio2}, and the right hand side of~\eqref{eq:ineqmain} is equivalent to the denominator of~\eqref{eq:ratio2}. Furthermore, if we sum both sides of inequality~\eqref{eq:ineqmain} across all steps of \textsf{Pivot}, we have
	\begin{align}
		\label{eq:almost} 
		|B^\texttt{++}| + |I^{\texttt{-}\texttt{-}}| \leq 2|I^\texttt{+-}| + \sum_{uv \in B^\texttt{++} \cup I^{\texttt{-}\texttt{-}}} 3 x_{uv}.
	\end{align}
	
	For every $uv \in E^\texttt{+-} = E_\mathcal{D}$, the LP constraints guarantee that $x_{uv} + x_{\rho(uv)} \geq 1$. Furthermore, if $uv$ is inside of a cluster, then $\mathcal{C}$ must make a mistake at $\rho(uv)$, or else there would be a hostile edge violation. Therefore,
	\begin{align}
		\label{eq:dbound}
		|I^\texttt{+-}| \leq |E^\texttt{+-}|/2 \leq \sum_{uv \in E^\texttt{+-}} x_{uv}.
	\end{align}
Combining this with~\eqref{eq:almost} allows us to bound $\text{cost}_G(\mathcal{C}) = |B^\texttt{+-}| + |B^\texttt{++}| + |I^{\texttt{-}\texttt{-}}|$ as follows
\begin{align*}
|B^\texttt{+-}| + |B^\texttt{++}| + |I^{\texttt{-}\texttt{-}}|  \leq |B^\texttt{+-}| + 2|I^\texttt{+-}| + \sum_{uv \in B^\texttt{++} \cup I^{\texttt{-}\texttt{-}}} 3 x_{uv} \leq \frac{3}{2} |E^\texttt{+-}| + \sum_{uv \in B^\texttt{++} \cup I^{\texttt{-}\texttt{-}}} 3 x_{uv} \leq \sum_{uv \in E^+ \cup E^-} 3 x_{uv}.
\end{align*}
Since we assumed $\{x_{uv}\}_{uv \in E^+ \cup E^-}$ is a $(1+\varepsilon/3)$-approximate LP solution, this proves that the deterministic pivoting strategy returns a $(3+\varepsilon)$-approximate solution for \hcc{}. If we instead chose pivots uniformly at random, then inequality~\eqref{eq:ineqmain} holds in expectation at each step, as does the final bound on $\text{cost}_G(\mathcal{C})$. Since we do not need to explicitly solve the LP at any point to choose pivots uniformly at random, we can assume without loss of generality that $\{x_{uv}\}_{uv \in E^+ \cup E^-}$ represents an implicit set of optimal LP variables (i.e., $\varepsilon = 0$), in which case we see the expected cost is within a factor $3$ of the optimal \hcc{} solution.
\end{proof}

\section{Conclusion and Discussion}
We have presented a $(3+\varepsilon)$-approximation algorithm for \ccc{} with a runtime of $\tilde{O}(n^3)$ time, settling the open question of Fischer et al.~\cite{fischer2025faster} up to an arbitrarily small constant $\varepsilon > 0$. The algorithm is simple in that it applies a familiar pivoting framework and can be made purely combinatorial. It also relies on rounding a covering LP relaxation, rather than on a variant of the Canonical LP or the Cluster LP, which are more expensive and complicated to solve or approximate.

Our work also includes simpler algorithms when there are only hostile or only friendly constraints. Each simplified algorithm includes a key advance needed to generalize the celebrated 3-approximate \textsf{Pivot} algorithm to a case involving one of the new constraint types. For \fcc{}, the main advance is to incorporate superedge constraints into a new covering LP formulation. The rest follows mostly from a straightforward application of the techniques of van Zuylen and Williamson~\cite{vanzuylen2009deterministic}. For \hcc{}, a key step is to note that the algorithm will make a large number of mistakes at a certain class of edges ($E_\mathcal{D}$).
This leads to a tighter analysis than would be obtained by directly applying techniques of van Zuylen and Williamson (e.g., applying Theorem 3.1 of~\cite{vanzuylen2009deterministic}) as a black box. We remark that our strategy for \hcc{} directly generalizes the analysis for a recent 3-approximation for \textsc{Cluster Deletion}~\cite{balmaseda2024combinatorial}, which can be viewed as an instance of \ccc{} where all negative edges are hostile. 

The $(3+\varepsilon)$-approximation for \ccc{} combines the key advances used by algorithms for \fcc{} and \hcc{}. However, by themselves these advances appear insufficient to produce a $(3+\varepsilon)$-approximation for \ccc{}. In particular, the full algorithm for \ccc{} also uses the notion of \heap{} constraints, to ensure that a key inequality (see Lemma~\ref{lem:badtri}) holds for certain bad triangles. Although $O(n^3)$ \heap{} constraints are sufficient to guarantee the $(3+\varepsilon)$-approximation, an open question is whether they are strictly necessary to obtain this approximation, or if a different proof strategy could avoid them. Finally, it is interesting to note that for \hcc{}, we can obtain an $O(n^3)$-time 3-approximation, where the approximation does not depend on $\varepsilon > 0$, the runtime does not involve logarithmic factors, and there is no need to explicitly solve an LP relaxation. An interesting question is whether this could also be achieved for \ccc{}, or even just for \fcc{}, using alternative techniques.

\section*{Acknowledgments}
Nate Veldt is supported by  AFOSR Award Number FA9550-25-1-0151 and by ARO Award Number W911NF-24-1-0156. The views and conclusions contained in this document are those of the authors and should not be interpreted as representing the official policies,
either expressed or implied, of the Army Research Office or the U.S. Government.

\bibliography{refs}

\newcommand{\etalchar}[1]{$^{#1}$}
\begin{thebibliography}{CCAL{\etalchar{+}}25}

\bibitem[ACN08]{AilonCharikarNewman2008}
Nir Ailon, Moses Charikar, and Alantha Newman.
\newblock Aggregating inconsistent information: ranking and clustering.
\newblock {\em Journal of the ACM}, 55(5):23, 2008.

\bibitem[AZO19]{allen2019nearly}
Zeyuan Allen-Zhu and Lorenzo Orecchia.
\newblock Nearly linear-time packing and covering {LP} solvers: Achieving
  width-independence and-convergence.
\newblock {\em Mathematical Programming}, 175(1):307--353, 2019.

\bibitem[BBC04]{bansal2004correlation}
Nikhil Bansal, Avrim Blum, and Shuchi Chawla.
\newblock Correlation clustering.
\newblock {\em Machine Learning}, 56:89--113, 2004.

\bibitem[BD08]{bhattacharya2008divisive}
Anindya Bhattacharya and Rajat~K De.
\newblock Divisive correlation clustering algorithm (dcca) for grouping of
  genes: detecting varying patterns in expression profiles.
\newblock {\em Bioinformatics}, 24(11):1359--1366, 2008.

\bibitem[BDY99]{ben1999clustering}
Amir Ben-Dor and Zohar Yakhini.
\newblock Clustering gene expression patterns.
\newblock In {\em International Conference on Computational Molecular Biology},
  pages 33--42, 1999.

\bibitem[BGSG22]{bonchi2022correlation}
Francesco Bonchi, David Garc{\'\i}a-Soriano, and Francesco Gullo.
\newblock {\em Correlation Clustering: Morgan \& Claypool Publishers}.
\newblock Morgan \& Claypool Publishers, 2022.

\bibitem[BV23]{bengali2023faster}
Vedangi Bengali and Nate Veldt.
\newblock Faster approximation algorithms for parameterized graph clustering
  and edge labeling.
\newblock In {\em International Conference on Information and Knowledge
  Management}, CIKM, 2023.

\bibitem[BXCV24]{balmaseda2024combinatorial}
Vicente Balmaseda, Ying Xu, Yixin Cao, and Nate Veldt.
\newblock Combinatorial approximations for cluster deletion: Simpler, faster,
  and better.
\newblock In {\em International Conference on Machine Learning}, ICML, 2024.

\bibitem[CALLN23]{cohen2023handling}
Vincent Cohen-Addad, Euiwoong Lee, Shi Li, and Alantha Newman.
\newblock Handling correlated rounding error via preclustering: A
  1.73-approximation for correlation clustering.
\newblock In {\em Symposium on Foundations of Computer Science}, FOCS, 2023.

\bibitem[CALN22]{cohen2022correlation}
Vincent Cohen-Addad, Euiwoong Lee, and Alantha Newman.
\newblock Correlation clustering with sherali-adams.
\newblock In {\em Symposium on Foundations of Computer Science}, FOCS, 2022.

\bibitem[CALP{\etalchar{+}}24]{cohen2024combinatorial}
Vincent Cohen-Addad, David~Rasmussen Lolck, Marcin Pilipczuk, Mikkel Thorup,
  Shuyi Yan, and Hanwen Zhang.
\newblock Combinatorial correlation clustering.
\newblock In {\em Symposium on Theory of Computing}, STOC, 2024.

\bibitem[CCAL{\etalchar{+}}24]{cao2024understanding}
Nairen Cao, Vincent Cohen-Addad, Euiwoong Lee, Shi Li, Alantha Newman, and
  Lukas Vogl.
\newblock Understanding the cluster linear program for correlation clustering.
\newblock In {\em Symposium on Theory of Computing}, STOC, 2024.

\bibitem[CCAL{\etalchar{+}}25]{cao2025solving}
Nairen Cao, Vincent Cohen-Addad, Euiwoong Lee, Shi Li, David~Rasmussen Lolck,
  Alantha Newman, Mikkel Thorup, Lukas Vogl, Shuyi Yan, and Hanwen Zhang.
\newblock Solving the correlation cluster {LP} in sublinear time.
\newblock In {\em Symposium on Theory of Computing}, STOC, 2025.

\bibitem[CGW05]{CharikarGuruswamiWirth2005}
Moses Charikar, Venkatesan Guruswami, and Anthony Wirth.
\newblock Clustering with qualitative information.
\newblock {\em Journal of Computer and System Sciences}, 71(3):360 -- 383,
  2005.

\bibitem[CHS24]{cao2024breaking}
Nairen Cao, Shang-En Huang, and Hsin-Hao Su.
\newblock Breaking 3-factor approximation for correlation clustering in
  polylogarithmic rounds.
\newblock In {\em Symposium on Discrete Algorithms}, SODA, 2024.

\bibitem[CMSY15]{ChawlaMakarychevSchrammEtAl2015}
Shuchi Chawla, Konstantin Makarychev, Tselil Schramm, and Grigory Yaroslavtsev.
\newblock Near optimal {LP} rounding algorithm for correlation clustering on
  complete and complete k-partite graphs.
\newblock In {\em Symposium on Theory of Computing}, STOC, 2015.

\bibitem[FKKT25a]{fischer2025faster}
Nick Fischer, Evangelos Kipouridis, Jonas Klausen, and Mikkel Thorup.
\newblock A faster algorithm for constrained correlation clustering.
\newblock In {\em Symposium on Theoretical Aspects of Computer Science}, STACS,
  2025.

\bibitem[FKKT25b]{fischer2025faster_arxiv}
Nick Fischer, Evangelos Kipouridis, Jonas Klausen, and Mikkel Thorup.
\newblock A faster algorithm for constrained correlation clustering.
\newblock {\em arXiv preprint arXiv:2501.03154}, 2025.

\bibitem[Fle04]{fleischer2004fast}
Lisa Fleischer.
\newblock A fast approximation scheme for fractional covering problems with
  variable upper bounds.
\newblock In {\em Symposium On Discrete Algorithms}, SODA, 2004.

\bibitem[GK98]{garg1998faster}
Naveen Garg and Jochen K{\"o}nemann.
\newblock Faster and simpler algorithms for multicommodity flow and other
  fractional packing problems.
\newblock In {\em Symposium on Foundations of Computer Science}, FOCS, 1998.

\bibitem[HCLM09]{hassanzadeh2009framework}
Oktie Hassanzadeh, Fei Chiang, Hyun~Chul Lee, and Ren{\'e}e~J Miller.
\newblock Framework for evaluating clustering algorithms in duplicate
  detection.
\newblock {\em Proceedings of the VLDB Endowment}, 2(1):1282--1293, 2009.

\bibitem[KKV25]{kalavas2025towards}
Andreas Kalavas, Evangelos Kipouridis, and Nithin Varma.
\newblock Towards better-than-2 approximation for constrained correlation
  clustering.
\newblock In {\em International Conference on Machine Learning}, ICML, 2025.

\bibitem[KNKY11]{kim2011highcc}
Sungwoong Kim, Sebastian Nowozin, Pushmeet Kohli, and Chang~D. Yoo.
\newblock Higher-order correlation clustering for image segmentation.
\newblock In {\em Advances in Neural Information Processing Systems}, NeurIPS,
  2011.

\bibitem[LDPM17]{Li2017motifcc}
P.~Li, H.~Dau, G.~Puleo, and O.~Milenkovic.
\newblock Motif clustering and overlapping clustering for social network
  analysis.
\newblock In {\em IEEE Conference on Computer Communications}, INFOCOM, 2017.

\bibitem[MC23]{makarychev2023single}
Konstantin Makarychev and Sayak Chakrabarty.
\newblock Single-pass pivot algorithm for correlation clustering. keep it
  simple!
\newblock In {\em Advances in Neural Information Processing Systems}, NeurIPS,
  2023.

\bibitem[RVG20]{ruggles2020parallel}
Cameron Ruggles, Nate Veldt, and David~F. Gleich.
\newblock A parallel projection method for metric constrained optimization.
\newblock In {\em Proceedings of the SIAM Workshop on Combinatorial Scientific
  Computing}, SIAM CSC, 2020.

\bibitem[SG22]{sonthalia2022project}
Rishi Sonthalia and Anna~C Gilbert.
\newblock Project and forget: solving large-scale metric constrained problems.
\newblock {\em Journal of Machine Learning Research}, 23(326):1--54, 2022.

\bibitem[Vel22]{veldt2022correlation}
Nate Veldt.
\newblock Correlation clustering via strong triadic closure labeling: Fast
  approximation algorithms and practical lower bounds.
\newblock In {\em International Conference on Machine Learning}, ICML, 2022.

\bibitem[VGW18]{veldt2018correlation}
Nate Veldt, David~F. Gleich, and Anthony Wirth.
\newblock A correlation clustering framework for community detection.
\newblock In {\em World Wide Web Conference}, WWW, 2018.

\bibitem[VGWS19]{veldt2019metric}
Nate Veldt, David~F. Gleich, Anthony Wirth, and James Saunderson.
\newblock Metric-constrained optimization for graph clustering algorithms.
\newblock {\em SIAM Journal on Mathematics of Data Science}, 1(2):333--355,
  2019.

\bibitem[vZW09]{vanzuylen2009deterministic}
Anke van Zuylen and David~P. Williamson.
\newblock Deterministic pivoting algorithms for constrained ranking and
  clustering problems.
\newblock {\em Mathematics of Operations Research}, 34(3):594--620, 2009.

\bibitem[WRM16]{wang2016unified}
Di~Wang, Satish Rao, and Michael~W Mahoney.
\newblock Unified acceleration method for packing and covering problems via
  diameter reduction.
\newblock In {\em International Colloquium on Automata, Languages, and
  Programming}, ICALP, 2016.

\bibitem[YIF12]{yarkony2012fast}
Julian Yarkony, Alexander Ihler, and Charless~C. Fowlkes.
\newblock Fast planar correlation clustering for image segmentation.
\newblock In {\em The European Conference on Computer Vision}, ECCV, 2012.

\bibitem[YSM{\etalchar{+}}24]{yu2024parclusterers}
Shangdi Yu, Jessica Shi, Jamison Meindl, David Eisenstat, Xiaoen Ju, Sasan
  Tavakkol, Laxman Dhulipala, Jakub \L{}\k{a}cki, Vahab Mirrokni, and Julian
  Shun.
\newblock The parclusterers benchmark suite (pcbs): A fine-grained analysis of
  scalable graph clustering.
\newblock {\em Proc. VLDB Endow.}, 18(3):836–849, November 2024.

\end{thebibliography}
\bibliographystyle{alpha}

\end{document}